\documentclass[12pt,reqno]{amsart}

\usepackage{lipsum}
\usepackage{amsfonts}
\usepackage{graphicx}
\usepackage{epstopdf}
\usepackage{algorithmic}
\ifpdf
  \DeclareGraphicsExtensions{.eps,.pdf,.png,.jpg}
\else
  \DeclareGraphicsExtensions{.eps}
\fi

\usepackage{amsopn}

\usepackage[left=3cm,top=2cm,right=3cm,bottom=2cm]{geometry}
\usepackage{amssymb}
\usepackage{epsfig}
\usepackage{tikz}
\usepackage[english]{babel}    % para texto em InglÃªs
\usepackage{enumitem}
\usepackage{multicol}
\usepackage{blkarray}
\usepackage{tikz}
\usepackage{amsmath}
\usepackage{mathtools}
\usepackage{float}
\usepackage{bm}
\usepackage{nicematrix}
\usepackage{comment}

\usepackage{booktabs}
\usepackage{systeme}

\usepackage{graphicx,subcaption} 
\DeclareMathOperator{\tr}{tr}
\DeclareMathOperator{\sinc}{sinc}
\DeclarePairedDelimiterX{\norm}[1]{\lVert}{\rVert}{#1}

\DeclareMathOperator*{\argmin}{arg\,min}

\newtheorem{theorem}{Theorem}[section]
\newtheorem{state}{Statement}[section]
\newtheorem{example}{Example}[section]
\newtheorem{lemat}{Lemma}[section]
\newtheorem{propt}{Property}[section]
\usepackage{algorithm}

\usepackage{etoolbox}

\begin{document}
\title[Multi-Community Spectral Clustering for Geometric Graphs]{Multi-Community Spectral Clustering for Geometric Graphs}

\author[L. E. Allem]{Luiz Emilio Allem}
 \address{UFRGS, Instituto de Matem\'atica e Estatística, Porto Alegre, Brazil}\email{emilio.allem@ufrgs.br}

\author[K. Avrachenkov]{Konstantin Avrachenkov}
\address{INRIA, NEO Team, Sophia Antipolis - Méditerranée}\email{k.avrachenkov@inria.fr}

\author[C. Hoppen]{Carlos Hoppen}
\address{UFRGS, Instituto de Matem\'atica e Estatística, Porto Alegre, Brazil}\email{choppen@ufrgs.br}

\author[H. Manjunath]{Hariprasad Manjunath}
 \address{INRIA, NEO Team, Sophia Antipolis - Méditerranée}\email{mhariprasadkansur@gmail.com}

\author[L. S. Sibemberg]{Lucas Siviero Sibemberg}
 \address{UFRGS, Instituto de Matem\'atica e Estatística, Porto Alegre, Brazil}\email{lucas.siviero@ufrgs.br}

\maketitle

% REQUIRED
\begin{abstract}
In this paper, we consider the soft geometric block model (SGBM) with a fixed number $k \geq 2$ of homogeneous communities in the dense regime, and we introduce a spectral clustering algorithm for community recovery on graphs generated by this model. Given such a graph, the algorithm produces an embedding into $\mathbb{R}^{k-1}$ using the eigenvectors associated with the $k-1$ eigenvalues of the adjacency matrix of the graph that are closest to a value determined by the parameters of the model. It then applies $k$-means clustering to the embedding. We prove weak consistency and show that a simple local refinement step ensures strong consistency. A key ingredient is an application of a non-standard version of Davis–Kahan theorem to control eigenspace perturbations when eigenvalues are not simple. We also analyze the limiting spectrum of the adjacency matrix, using a combination of combinatorial and matrix techniques. 

\end{abstract}

\section{Introduction and Main Results}

The history of science has been marked by attempts to make sense of data and measurements and to explain them in a sensible way. A natural step in this direction is to organize the data in a (hopefully small) number of groups that somehow capture the main features of its objects. Objects with similar characteristics must belong to the same group, while dissimilar objects must be placed in separate groups. Quoting Jain~\cite{JainReview}, ``cluster analysis is the formal study of methods and algorithms for grouping, or clustering, objects according to measured or perceived intrinsic characteristics or similarity.'' In opposition to classification or discriminant analysis (supervised learning), for which objects are tagged with class labels defined by an external source, data clustering aims to assign the objects to classes that are not defined a priori, and is supposed to capture intrinsic properties or the underlying structure of the data set. 
%This is why clustering is known as unsupervised learning, while discriminant analysis is a type of supervised learning. 
Algorithms based on eigenvalues and eigenvectors play a prominent role in uncovering complex dependencies in a data set.

%Traditional clustering methods such as single linkage~\cite{1973sibson} and $k$-means~\cite{1967macqueen}, although simple to implement, often struggle with complex data structures, such as those found in graphs. For instance and Single Linkage is sensitive to outliers, $k$-means assumes convex clusters. In this scenery, the spectral clustering stands out as a powerful technique~\cite{2007luxburg}. The foundation of spectral clustering lies in linear algebra; it uses the eigenvalues and eigenvectors of a matrix, such as the graph’s Laplacian, which encodes the similarity between data points. This approach enables the identification of non-convex structures, making it suitable for more complex clustering tasks.

Data clustering is widely used in real-world applications in areas such as biology~\cite{2010guimin}, computer science~\cite{2021wei}, economics~\cite{2021mansano}, medicine~\cite{PR2023} and social sciences~\cite{fortunato2010community}. In parallel, there is a large body of work related to the design and analysis of clustering algorithms~\cite{1967macqueen,Newman06,2001ng,1973sibson}. Success is often based on the algorithm's ability to recover ``the ground truth'' of an artificial data set or to achieve high agreement with the classifications of benchmark data sets. Von Luxburg, Williamson, and Guyon~\cite{LWG12} discuss practical and epistemological difficulties of context-free evaluation of clustering algorithms and argue for a more problem-dependent approach and for a systematic catalog of clustering problems. As Jain~\cite{JainReview} puts it, ``a cluster is a subjective entity that is in the eye of the beholder and whose significance and interpretation require domain knowledge''.

According to~\cite{LWG12}, an environment that greatly simplifies real-world data sets and where the aim of cluster analysis can be made precise is that of constraint-based models that assume interactions between samples, which lend themselves to graph partitioning methods. For example, Spielman and Teng~\cite{ST07} showed that spectral partitioning methods based on the eigenvector associated with the algebraic connectivity work well on bounded-degree planar graphs and finite element meshes. Lee, Gharan, and Trevisan~\cite{LGT14} provided a theoretical justification for algorithms that use the eigenspaces associated with the bottom $k$ eigenvalues of Laplacian matrices to embed data points in $\mathbb{R}^k$, and then cluster these points based on geometric considerations. Von Luxburg, Belkin, and Bousquet~\cite{LBB08} have obtained results about the consistency of spectral clustering. Under some mild assumptions, they have shown that clusterings constructed by Laplacian-based spectral clustering algorithms converge almost surely to a limit clustering of the entire data space. %Lei and Rinaldo~\cite{lei2015consistency} obtained consistency results for spectral clustering applied to the adjacency matrix, and showed that it can be used to identify hidden communities in stochastic block models.

A very natural random graph model with an underlying structure is the Stochastic Block Model (SBM, for short), which was introduced by Holland, Laskey, and Leinhardt~\cite{HLL83}. Given a number of nodes $n$ and a number of communities $k$, an initial partition is given, or, alternatively, each node is initially assigned uniformly at random to one of the communities. Next, for any two nodes $i$ and $j$, an edge $\{i,j\}$ is drawn, independently from the other edges, with some probability $p_{i,j}$ that only depends on the communities of $i$ and $j$. 
%This model has several applications to real world situations [?] and is the most studied one~\cite{lei2015consistency}.
Clustering in such a graph corresponds to the inverse problem where one wishes to extract the $k$ communities from a graph $G$ that was generated using SBM. Lei and Rinaldo~\cite{lei2015consistency} showed that spectral clustering leads to perfect extraction under reasonably mild conditions, also for rather sparse regimes. 
%And, when considering two communities, other methods lead to perfect recovery even in sparse regimes, where the number of neighbors of each node is $\Theta(\log n)$~\cite{galhotra2018geometric}.
We refer the interested reader to Abbe~\cite{abbe2018community} for a survey of related results. 

In most practical situations, nodes would typically have other attributes beyond the community label. For instance, spatial attributes may be captured by an embedding in a metric space. In such geometric models, the connection between a pair of nodes depends both on their communities and on their relative positions in the metric space. Models of this type may be classified as \emph{geometric block models} if the embedding of the nodes into the metric space is random, but the criterion for drawing an edge is deterministic based on their locations, or as \emph{soft geometric block models} if the locations of the nodes define a probability distribution for the edges. In both cases, given three nodes $i$, $j$, and $\ell$, the event that $i$ and $\ell$ are adjacent is not independent from the events that $i$ and $j$, or $j$ and
$\ell$, are adjacent. Community detection has been explored for geometric block models~\cite{galhotra2018geometric,galhotra2023community}, for Euclidean random geometric graphs~\cite{10.1093/imaiai/iaaa009,Avrachenkov2024Euclidean,GaudioGuan2025,GaudioNW24,sankararaman2018community} and for the soft geometric block model of Avrachenkov, Bobu, and Dreveton~\cite{avrachenkov2022}. 
%So, it is more likely that there is an edge between $i$ and $\ell$ if $\{i,j\}$ and $\{j,\ell\}$ are two edges. This transitivity rule approximate these models to more real-world settings, like social networks, blog-networks and advertising [?]. 
This body of work shows that there are methods that can successfully identify the community structure in (soft) geometric block models. Nevertheless, direct applications of classical spectral clustering algorithms, which consider eigenvectors associated with the top or bottom eigenvalues of the corresponding matrices, often fail. This is to be expected, as classical algorithms seek a classification such that the elements in the same class are all similar to each other, while elements in different classes are dissimilar. In terms of graphs, this typically means that elements in the same class tend to be joined to each other by short paths. However, the dependence on the geometry may force elements of the same community to be far from each other.
%say if the criterion for adding an edge is based on the distance between the random points associated with each embedding. 
To illustrate the difference with an informal example, suppose that we have a graph such that the nodes are people and the edges tell us when two people are friends. Classical algorithms would likely sort people according to where they live, while the underlying community structure might actually sort people according to generation if it is true that people are more likely to have friends of a similar age.   

In 2022, the authors of~\cite{avrachenkov2022} considered a soft geometric block model with two communities and showed that the communities may be perfectly recovered using a spectral algorithm. Crucially, the algorithm does not necessarily use one of the top or bottom eigenvalues. The main objective of this paper is to generalize their model to any fixed number $k \geq 2$ of communities, which requires ingredients of random matrix theory and the control of eigenspace perturbations. To describe the results in~\cite{avrachenkov2022} and our contributions,  we conclude the introduction with a description of the model and with an informal account of the results that aims to convey their meaning in a non-technical way. Formal statements and definitions are deferred to Section~\ref{sec:proof}.
 
\subsection{Soft Geometric Block Model}\label{sec:sgbm}
The model in~\cite{avrachenkov2022}, which was called the Soft Geometric Block Model (SGBM), generalizes both the stochastic block model and the geometric block model in~\cite{galhotra2018geometric} as well as Euclidean random matrices \cite{bordenave2008}. %They considered the community detection problem under some additional assumptions on the parameters. 
% {\color{red}Bordenave, C. (2008). Eigenvalues of Euclidean random matrices. Random Structures \& Algorithms, 33(4), 515-532.}
Their model is defined in a compact and homogeneous metric space, the $d$-dimensional flat unit torus $\mathbf{T}^d = \mathbb{R}^d/\mathbb{Z}^d$. Let $D=[n]=\{1,\ldots,n\}$ be a set of $n$ points, and let $K = [k] = \{ 1,2, \ldots, k\}$ be a set of communities. Consider a community assignment $\sigma:D\rightarrow K$ and an embedding $X:D\rightarrow \mathbf{T}^d$, where $\sigma_i=\sigma(i)$ denotes the community label of vertex $i$ and the $i$-th coordinate of $X = (X_1,X_2, \ldots , X_n)$ is the vector corresponding to $i$ in the metric space. Let $F\colon \mathbf{T}^d\times K\times K\rightarrow\mathbb{R}_+$ be a measurable nonnegative function such that $F(\cdot,\sigma_i,\sigma_j) = F(\cdot,\sigma_j,\sigma_i)$. According to this model, given $i,j\in [n]$ the edge $\{i,j\}$ appears with probability $F(X_i-X_j, \sigma_i,\sigma_j)$, where the function depends only on the distance 
\addtocounter{footnote}{1}
$\norm{X_i-X_j}$\footnote{Unless otherwise stated, the notation $\norm{\cdot}$ stands for the $\ell_\infty$-norm, that is $\norm{X}=\max |X_i|$.} and on the community labels of $i$ and $j$. More precisely, given $\sigma \colon D \rightarrow [k]$ and $X \colon D \rightarrow \mathbf{T}^d$, 
% {\color{red}$D$ or $V$? Maybe $V$ is better? on the other hand, we have $V$ as vectors...} 
the SGBM model defines the graph $G$ with $n$ nodes in terms of its $n \times n$ adjacency matrix $A=(a_{ij})=A(G)$, where $a_{ij}=a_{ji}=1$ if, and only if, $\{i,j\}$ is an edge of $G$. The distribution of the adjacency matrix is given by
\begin{equation}\label{MAdistribution}
    \mathbb{P}_{\sigma,X}(A) = \prod \limits_{1 \leq i < j \leq n} (1 - F(X_i-X_j,\sigma_i,\sigma_j))^{1-a_{ij}} (F(X_i-X_j,\sigma_i,\sigma_j))^{a_{ij}}. 
\end{equation}
Note that this model coincides with the SBM if $F$ does not depend on $X$, that is, $F(X,\sigma_i,\sigma_j)=p_{i,j}$. It is a GBM if there are $r_{in},r_{out}>0$ such that 
\begin{equation}\label{gbm}
F(X,\sigma_i,\sigma_j)=
\begin{cases}1, & \textrm{ if }\sigma_i=\sigma_j \textrm{ and } \norm{X_i-X_j} \leq r_{in},\\
1, & \textrm{ if }\sigma_i\neq \sigma_j \textrm{ and } \norm{X_i-X_j} \leq r_{out},\\
0, & \textrm{ otherwise}.
\end{cases}
\end{equation}

The community detection problem studied in~\cite{avrachenkov2022} may be stated as follows. For a fixed $n$, assume that a secret assignment $\sigma$ is chosen, that the node positions $X_i$ are chosen independently and with uniform probability in $\mathbf{T}^d$ (u.a.r.\ in $\mathbf{T}^d$ for short), and that a graph $G$ is chosen according to~\eqref{MAdistribution}. The aim is to find $\sigma$ using $G$. This corresponds to computing an estimator $\hat{\sigma}$ whose quality is measured as follows. Given $\sigma,\sigma' \colon D \rightarrow [k]$, we say that $\sigma$ and $\sigma'$ are equivalent if the codomain of $\sigma'$ may be relabeled in a way that turns it into $\sigma$. Formally, for any $i,j\in\{1,\ldots,n\}$, $\sigma_i=\sigma_j$ if, and only if, $\sigma'_i=\sigma'_j$. As usual, the Hamming distance is given by
\begin{equation}\label{hamming}
d_{H}(\sigma,\sigma') = |\{i\in n:\sigma_i\neq \sigma'_i\}|.
\end{equation}
We define the \emph{absolute classification error} and the \emph{loss rate} for an estimator $\hat{\sigma}$ of $\sigma$ as 
\begin{equation}\label{loss}
d^\ast_{H}(\sigma,\hat{\sigma}) = \min\{d_{H}(\sigma,\hat{\sigma}'):\hat{\sigma}'\text{ is equivalent to }\hat{\sigma}\} \textrm{ and }\ell(\sigma,\hat{\sigma})=\frac{d_H^\ast(\sigma,\hat{\sigma})}{n}.
\end{equation}

The authors of~\cite{avrachenkov2022} considered this community detection problem for $k=2$ with the following additional assumptions:
\begin{itemize}
\item[i)] the communities have equal size, that is, $\left|\{i\in[n]:\sigma_i=q\}\right|=n/k,$
for every $q\in[k]$;

%\item[ii)] community labels are randomly assigned, that is, the $n/k$ nodes with community label $\ell\in[k]$ are chosen with uniform probability among all available nodes;

%\item[iii)] the position $X_i$ of each node $v_i$ is uniformly distributed over $\mathbf{T}^d$;

\item[ii)] the function $F$ is defined for $x\in \mathbf{T}^d$ as 
\begin{equation}\label{Mfunction}
    F(x , \sigma_i , \sigma_j) = \begin{cases}
        F_{in}(x), \quad  \text{if } \sigma_i = \sigma_j \\
 F_{out}(x), \quad \text{otherwise,}\\
    \end{cases}
\end{equation}
where the functions $F_{in},F_{out} : \mathbf{T}^d \to [0,1]$ are measurable functions known as connectivity probability functions. 
\end{itemize}
Let $\mu_{in}$ and $\mu_{out}$ be the expected intracommunity and intercommunity edge densities, that is, a vertex is expected to have $\mu_{in}\left(\frac{n}{k}-1\right)$ neighbors within its own community and $\frac{\mu_{out}(k-1)n}{k}$ neighbors outside its community. 

%The expected intra-community edge density $\mu_{in}$ and inter-community edge density  are given by 
%\begin{equation} \mu_{in} = \int_{\mathbf{T}^d} F_{in}(x) dx \textrm{ and } \mu_{out} = \int_{\mathbf{T}^d} F_{out}(x) dx.\end{equation} 
%This means that 

%For a measurable function $F : \mathbf{T}^d \to \mathbb{R}$, and $z \in \mathbb{Z}^d$ we denote the Fourier transform  
%\begin{align*}
%    \hat{F}(z) = \int_{\mathbf{T}^d} F(z) e^{-2 \pi \texttt{i} \langle z,x \rangle} dx,
%\end{align*}
%The Fourier series is given by 
%\begin{align*}
%    F(x) = \sum_{\mathbb{Z}^d} \hat{F}(z) e^{2 \pi \texttt{i} \langle z,x \rangle}.
%\end{align*}

%For two integrable functions $F,G : \mathbf{T}^d \to \mathbb{R}$, we define the convolution $F * G (y) = \int_{\mathbf{T}^d} F(x) G(y-x) dx$. We recall that convolution in function space is product in Fourier space.

%In this context, clustering corresponds to the following inverse problem. Given the adjacency matrix $A$ and the function $F$, identify the community structure $\sigma$ that was used to generate $A$. Avrachenkov, Bobu, and Dreverton~\cite{avrachenkov2022} considered this problem for $k=2$ in the dense setting, i.e., in the case where $\mu_{in}+\mu_{out} \geq \varepsilon$ where $\varepsilon>0$ is an absolute constant.

We are now ready to state an informal version of the main result of~\cite{avrachenkov2022}. For a formal statement, see Section~\ref{sec:proof}. As usual, given a sequence of probability spaces
$(\Omega_i, \mathbf{P}_i)_{i\in \mathbb{N}}$, we say that a sequence of events $(A_i)_{i\in\mathbb{N}}$, where $A_i \subset \Omega_i$, holds \emph{asymptotically almost
  surely} (a.a.s.\ for short) if $\mathbf{P}_n(A_n)\to 1$ as $n\to\infty$. 
\begin{theorem}\cite{avrachenkov2022}\label{thm_main_avrachenkov}
Assume that $F$, $\mu_{in}$ and $\mu_{out}$ satisfy technical conditions given in terms of the coefficients of the Fourier series of $F$. Assume that $\mu_{in}>\mu_{out}>0$. Let $n$ be a large even number and let $\sigma$ be an assignment of two communities of size $n/2$. Let $X_1,\ldots,X_n$ be chosen u.a.r.\ in $\mathbf{T}^d$. If $A$ is the adjacency matrix of a graph $G$ generated according to the SGBM with~\eqref{MAdistribution} and~\eqref{Mfunction}, then the following hold a.a.s.:
\begin{itemize}
\item[(a)] The eigenvalue $\lambda$ of $A$ that is closest to $n(\mu_{in}-\mu_{out})/2$ is simple and is `far' from any other eigenvalue of $A$. 

\item[(b)] Any eigenvector of $A$ associated with $\lambda$ produces an estimator $\hat{\sigma}$ such that $\ell(\sigma,\hat{\sigma})=o(1)$.

\item[(c)] Assume that $\sigma'$ is the perturbation of $\hat{\sigma}$ obtained as follows: for each $i$, define $\sigma'_i=m$ if most neighbors $j$ of $i$ in $G$ satisfy $\hat{\sigma}_j=m$. Then $\ell(\sigma,\sigma')=0$.

\end{itemize}
\end{theorem}
Theorem~\ref{thm_main_avrachenkov} states that, under some technical conditions, the two communities that define an SGBM may be fully recovered from an eigenvector associated with a particular eigenvalue $\lambda$ of $A$. 
%As it turns out, we may fully recover $\sigma$ if we assume that $\mu_{in}>\mu_{out}$ and we modify $\hat{\sigma}$ by a majority vote, that is, we consider the estimator $\sigma'$ such that $\sigma'(i)=c$ if at least half the neighbors $v_j$ of $v_i$ satisfy $\hat{\sigma(j)}=c$. 

To better describe our contribution, we briefly describe the proof of Theorem~\ref{thm_main_avrachenkov} in~\cite{avrachenkov2022}. First, the authors used Talagrand's inequality and the Borel-Cantelli Lemma to show that the spectral measure associated with the (normalized) adjacency matrix of an $n$-vertex graph defined by the soft geometric block model converges in distribution to a limiting measure $\mu$ on $\mathbb{R}$. This is a discrete measure composed of two terms, one corresponding to a random graph with no community structure, and the other carrying information about the difference between intracommunity and intercommunity connection probabilities. 

The second step was to show that the following holds for a particular point $\tilde{\lambda}$ in the support of $\mu$, which has the property that $n\tilde{\lambda}$ is an eigenvalue of the matrix of expected connection probabilities. The spectrum of an $n$-vertex adjacency matrix $A$ selected according to the SGBM a.a.s.\ contains an eigenvalue $\lambda$ such that $|\lambda-n\tilde{\lambda}|=o(n)$ and there is $\varepsilon>0$ such that $|\lambda'-\lambda|\geq \varepsilon n$ for all remaining eigenvalues $\lambda'$ of $A$. The proof uses Fourier analysis and relies on the technical conditions in the statement of the theorem. The second step implies that $\lambda$ is a simple eigenvalue, so that there is essentially a unique unit eigenvector\footnote{Being precise, there are exactly two unit eigenvectors that only differ by their sense.} associated with it. In an algorithmic perspective, it shows that a computer correctly identifies $\lambda$ in the spectrum of $A$ with floating-point arithmetic. This gave part~(a).

The third step established (b) and consisted of showing that the eigenvector of $A$ associated with $\lambda$ a.a.s.\ classifies the data set into two clusters with loss rate $O((\log{n})/n)$. To prove this, the authors showed that this eigenvector a.a.s.\ forms a small angle with the eigenvector associated with $n\tilde{\lambda}$ with respect to the matrix of expected connection probabilities. Note that, because it is assumed that the points are embedded u.a.r.\ in the probability space and because one is taking expected values, the entries of this matrix of expected probabilities do not depend on the geometry and, therefore, the problem is reduced to SBM. The final step is simple, and results in the perfect recovery of the partition stated in part (c) after an additional local improvement step. 

The main contribution of this paper is an extension of Theorem~\ref{thm_main_avrachenkov} to arbitrary values of $k$. It is stated informally below. The statement refers to \emph{$k$-means clustering}, a simple iterative procedure introduced by MacQueen~\cite{1967macqueen} to cluster data embedded in a metric space. Assume that the aim is to cluster the data points into $\ell$ parts. It starts with an initial partition (say, a random partition) of the data points into $\ell$ parts. In subsequent steps, it computes the centroids of the points in each of the $\ell$ parts and updates the partition so that every point is assigned to the part whose centroid is closest to it. The procedure ends when the partition remains the same after the updating step. Since it is simple and easy to implement, $k$-means is a very popular clustering procedure. However, it is not able to extract any information from the data set beyond the relative distances of the data points.  
On the other hand, this does not mean that metric procedures such as $k$-means are useless for complex data sets. Many spectral clustering algorithms may be viewed as a 2-step procedure. In the first step, the data points are mapped into an auxiliary metric space based on spectral considerations. The distribution of points in this metric space turns out to be adequate for metric procedures, which are used to obtain the partition in the second step. This is also the case here.   
\begin{theorem}\label{thm_main_informal}
Assume that $F$, $\mu_{in}$ and $\mu_{out}$ satisfy technical conditions given in terms of the coefficients of the Fourier series of $F$. Assume that $\mu_{in}>\mu_{out}>0$. Let $k \geq 2$ be fixed, let $n$ be a large number divisible by $k$ and let $\sigma$ be an assignment of $k$ communities of size $n/k$. Choose $X_1,\ldots,X_n$ u.a.r.\ in $\mathbf{T}^d$.  If $A$ is the adjacency matrix of a graph $G$ generated according to the SGBM with~\eqref{MAdistribution} and~\eqref{Mfunction}, then the following hold a.a.s.:
\begin{itemize}
\item[(a)] The $k-1$ eigenvalues $\lambda_1,\ldots,\lambda_{k-1}$ of $A$ (including multiplicity) that are closest to $n(\mu_{in}-\mu_{out})/k$ are `far' from any other eigenvalue of $A$. 

\item[(b)] Consider the $n \times (k-1)$ matrix $V$ whose columns are unit eigenvectors of $A$ associated with the eigenvalues $\lambda_1,\ldots,\lambda_{k-1}$ of part (a). Consider the embedding of the set $D$ into $\mathbb{R}^{k-1}$ that associates each vertex $i$ with the $i$-th row of $V$. An application of $k$-means clustering to these points produces an estimator $\hat{\sigma}$ such that $\ell(\sigma,\hat{\sigma})=o(1)$.

\item[(c)] Assume that $\sigma'$ is the perturbation of $\hat{\sigma}$ obtained as follows: for each $i$, $\sigma'_i=m$ if most neighbors $j$ of $i$ in $G$ satisfy $\hat{\sigma}_j=m$. Then $\ell(\sigma,\sigma')=0$.

\end{itemize}
\end{theorem}

As before, the proof may be viewed in four steps. The first two steps prove part (a) and are reminiscent of what was done in~\cite{avrachenkov2022}, with additional technical difficulties arising from the larger number of classes. The third step is very different. Since we cannot ensure that the $k-1$ eigenvalues in part (b) are simple, but only that the other eigenvalues are far from them, the choice of orthogonal basis for the eigenspace is no longer essentially unique, and we must show that the procedure works for any possible orthogonal basis of eigenvectors. Furthermore, we must understand the effect of an application of $k$-means on the embeddings of the points in $\mathbb{R}^k$. The main ingredient is a non-trivial application of the Davis-Kahan Theorem (see Theorem~\ref{Davis-Kahan}), a result that is often used to bound the distance between the subspace spanned by a family of eigenvectors and the subspace spanned by their sample versions. To achieve our results, we prove auxiliary results in matrix theory that may be of independent interest. The perfect recovery described in part (c) is easy to prove, and may be established just as in~\cite{avrachenkov2022}.
%In this work we will be considering a SGBM in the dense regime with $k$ number of clusters. We establish a new higher-order spectral clustering algorithm for an arbitrary number of clusters. We prove that this algorithm is weakly consistent on the SGBM and, with a small tuning, we establish a strong consistency algorithm on the SGBM with an arbitrary number of clusters. This generalizes what was done in~\cite{avrachenkov2022} since the method worked only in the case where $k=2$. We note that, in this methods, considering $k$ communities, there are $k-1$ eigenvectors that together can recover the hidden communities. We also provide numerical experiments, 
%which appear some topological structure with different number of clusters, which shows our method working even in cases where $k$ is large. 

We conclude the introduction with the algorithms suggested by parts (b) and (c) of Theorem~\ref{thm_main_informal}.
\begin{algorithm}[H]
\caption{Higher-Order Spectral Clustering for $k$ clusters}
\label{Cluster_Alg}
\begin{algorithmic}[1]
\REQUIRE Adjacency matrix $A$, number of clusters $k$, average intracluster and intercluster edge densities $\mu_{\text{in}}$ and $\mu_{\text{out}}$.
\ENSURE Node labelling $\Tilde{\sigma}\in\{1,2,\ldots,k\}^n$.

 \STATE{Let $\lambda'_1\geq\cdots\geq\lambda'_{k-1}$ be the eigenvalues of $A$ closest to $\lambda_* = \frac{\mu_{\text{in}}-\mu_{\text{out}}}{k}n$;}

\STATE{Let $v_1,\ldots,v_{k-1}$ be orthogonal unit eigenvectors of $A$ associated with the eigenvalues $\lambda'_1,\ldots,\lambda'_{k-1}$, respectively. Let $V=[v_1\cdots v_{k-1}]\in\mathbb{R}^{n\times(k-1)}$;}

 \STATE{Split the set $\{\mathbf{w}_1,\ldots,\mathbf{w}_n\}$ of rows of $V$ into $k$ clusters $P_1,\ldots,P_k$ via $k$-means;}

 \STATE{For every node $i\in\{1,\ldots,n\}$, let $\hat{\sigma} = \ell$, if $\mathbf{v}_i\in P_\ell$.}

\textbf{Return} node labelling $\hat{\sigma}$

\end{algorithmic}
\end{algorithm}

\begin{algorithm}[H]
\caption{Higher-Order Spectral Clustering with local improvement}
\label{Cluster_Alg2}
\begin{algorithmic}[1]
\REQUIRE Adjacency matrix $A$, number of clusters $k$, average intracluster and intercluster edge densities $\mu_{\text{in}}$ and $\mu_{\text{out}}$.
\ENSURE Node labelling $\hat{\sigma}\in\{1,2,\ldots,k\}^n$.

\STATE Let $\Tilde{\sigma}\in\{1,2,\ldots,k\}^n$ be the output of Algorithm~\ref{Cluster_Alg};

\FOR{ $i=1,\ldots, n$}

\STATE $\hat{\sigma}_i = \text{arg}\max_{\ell\in[k]} \sum_{j=1}^n1(\Tilde{\sigma}_j=\ell)$.
\ENDFOR

\textbf{Return} node labelling $\hat{\sigma}$

\end{algorithmic}
\end{algorithm}

It should be mentioned that the only difference between this algorithm and the classical spectral clustering algorithm is the choice of eigenvectors. Instead of choosing the $k-1$ eigenvectors closest to $\lambda_*$, the classical algorithm (see, for instance, Algorithm 1~\cite{lei2015consistency}) uses the eigenvectors associated with the $k$ largest eigenvalues. We provide an example for insight. 
\begin{example}\label{example1}
Consider the $1$-dimensional geometric block model (GBM), and assume that we have $k=4$ communities, each consisting of 250 members. These $n=1000$ points have been embedded u.a.r.\ in $S^1$ and we have produced a graph $G$ according to~\eqref{gbm} for $r_{in}=0.43$ and $r_{out}=0.11$. Let $\lambda_1 \geq \cdots \geq \lambda_{1000}$ be the eigenvalues associated with the adjacency matrix of $A$. Consider the $1000 \times 3$ matrix $V$ whose columns are unit eigenvectors $v_4$, $v_5$, and $v_6$ associated with the eigenvalues $\lambda_4\approx 163.37$, $\lambda_5\approx 162.75$, and $\lambda_6\approx 160.65$, respectively, which are the three eigenvalues closest to $\lambda_*=n(\mu_{in}-\mu_{out})/k=100\cdot 0.64/4=160$. We observe that $\lambda_3 \approx 181.94$ and $\lambda_7 \approx 97.41$, so that $|\lambda_i-\lambda_*|\geq 21.94$ for all $i \notin \{4,5,6\}$. In Figure~\ref{exp_intro}, each row of $V$ is mapped onto the corresponding point in $\mathbb{R}^3$, and we can see that the elements in each of the four communities are separated in a way that is suitable for $k$-means. Figure~\ref{exp_intro2} depicts what would happen if the eigenvectors in $V$ were replaced by the eigenvectors $v_2,v_3,v_4$ associated with the eigenvalues $\lambda_2,\lambda_3,\lambda_4$, respectively. The figure illustrates that, although one pair of communities is clearly separated from the other pair in a way that can be captured by $k$-means, the separation between classes within each pair is not evident. In fact, $k$-means fails to distinguish the two communities within each pair using this embedding. We observe that the standard spectral algorithm uses four eigenvectors to cluster into four classes, namely the eigenvector $v_1$ associated with the largest eigenvalue $\lambda_1$ is used along with $v_2$, $v_3$ and $v_4$. However, since a random graph generated by this model is ``almost regular'', the eigenvector $v_1$ associated with $\lambda_1$ is ``almost'' a multiple of $(1,\ldots,1)$, which makes it unhelpful for clustering.
\end{example}

\begin{figure}
\centering
\includegraphics[width=14cm]{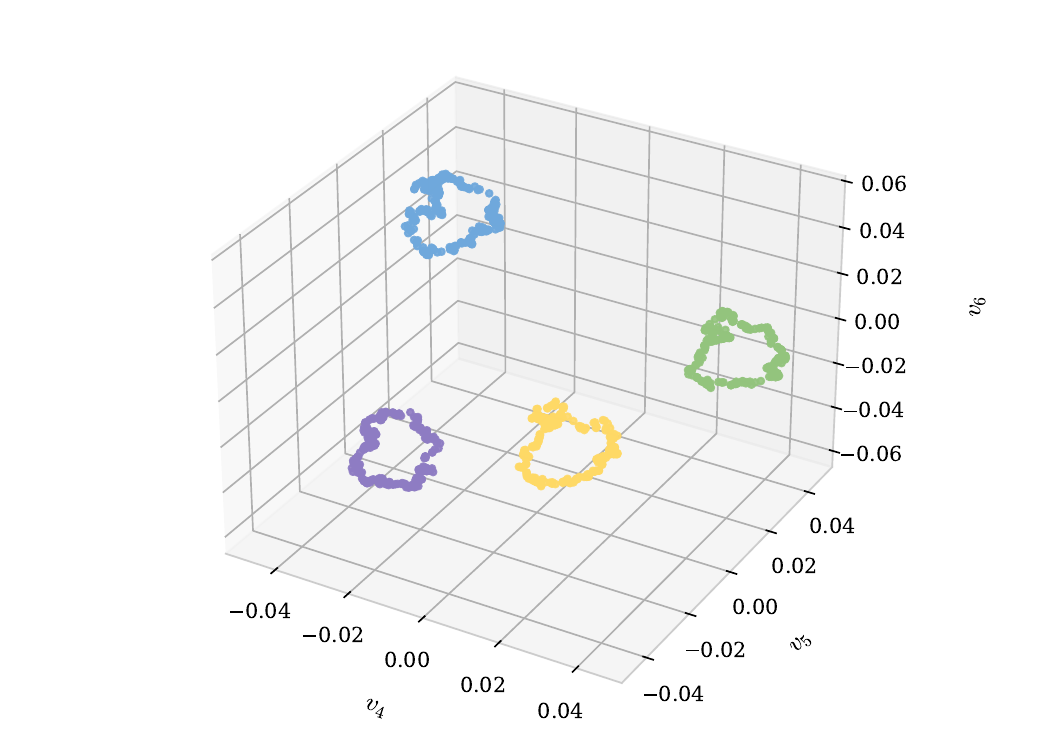}
\caption{The points in $\mathbb{R}^3$ that correspond to the elements in the setting of Example~\ref{example1} for the eigenvectors $v_4$, $v_5$ and $v_6$.
Elements in the same community have been drawn with the same color.}
\label{exp_intro}
\end{figure}

\begin{figure}
\centering
\includegraphics[width=14cm]{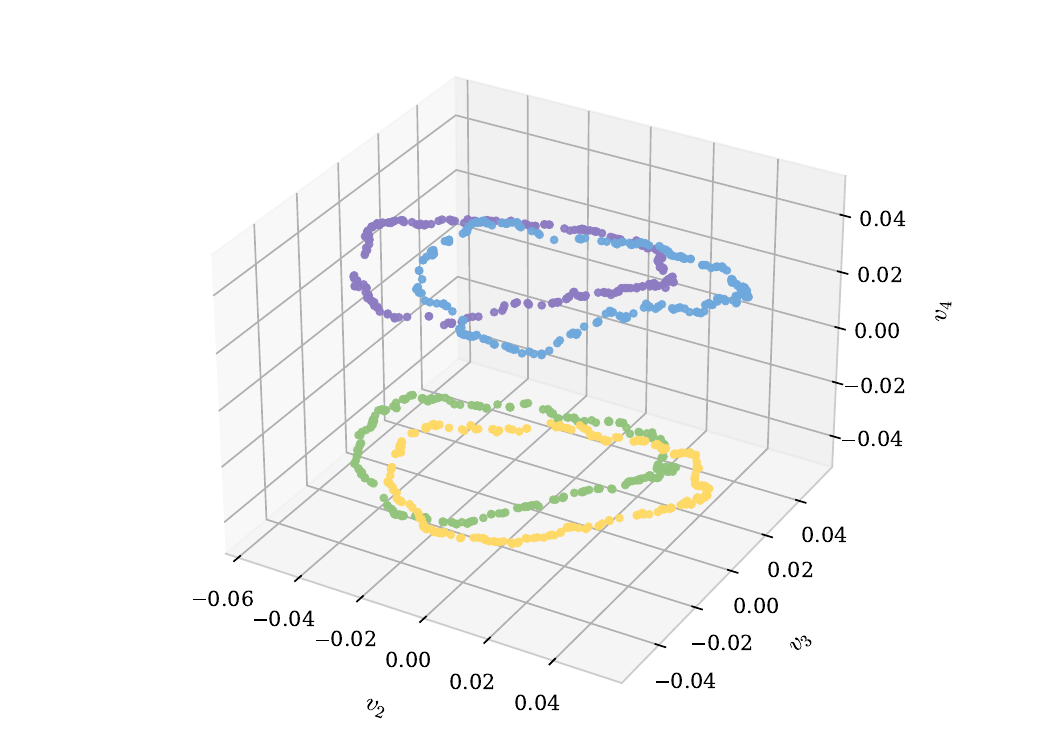}
\caption{The points in $\mathbb{R}^3$ that correspond to the elements in the setting of Example~\ref{example1} if the columns of $U$ were replaced by $v_2$, $v_3$ and $v_4$.
Elements in the same community have been drawn with the same color.}
\label{exp_intro2}
\end{figure}

The remainder of the paper is structured as follows. In Section~\ref{sec:proof}, we state our results formally, and give an overview of their proof. Understanding the limit distribution of the SGBM for $k \geq 2$ clusters is the subject of Section~\ref{sec:limit_spec}. Section~\ref{sec:daviskahan} contains the linear-algebraic results that are used to prove that the outputs of Algorithms~\ref{Cluster_Alg} and~\ref{Cluster_Alg2} a.a.s.\ satisfy the properties of Theorem~\ref{thm_main_informal}, which is the subject of Section~\ref{sec:consistency}. We conclude the paper with final remarks and open problems in Section~\ref{sec:simulation}.

%The remainder of the paper is organized as follows. 
%In Section \ref{sec:sgbm} we define the SGBM model that we are going to consider throughout the paper and some needed definitions like weakly and strong consistency. In Section \ref{sec:limit_spec} \ref{sec:main} we characterize the limiting spectrum and put a condition so we have some guarantees. We present the higher-order spectral clustering algorithm in Section \ref{sec:sec:daviskahan} providing the proof of its weakly consistency. In Section \ref{sec:consistency} we provide the adittional step that makes the algorithm strong consistency. Finally, we explore numerical simulations in Section \ref{sec:simulation} and we conclude the paper in Section~\ref{sec:concl} discussing future directions on this research. 

\section{Statement of the main result}\label{sec:proof}

%Given a graph $G$ with vertex set $D=[n]$, let $A=A(G)$ be the adjacency matrix and let $D=D(G)$ be the diagonal matrix of vertex degrees. 

Recall that $\mathbf{T}^d = \mathbb{R}^d/\mathbb{Z}^d$ is the $d$ dimensional flat unit torus. We follow the definition of a Soft Geometric Block Model given in Section~\ref{sec:sgbm}.

For a measurable function $\varphi \colon \mathbf{T}^d \to \mathbb{R}$, we consider its Fourier transform $\hat{\varphi} \colon \mathbb{Z}^d \to \mathbb{C}$ defined as 
\begin{align*}
    \hat{\varphi}(z) = \int_{\mathbf{T}^d} \varphi(x) e^{-2 \pi \texttt{i} \langle z,x \rangle} dx,
\end{align*}
where $\langle z,x \rangle$ denotes the usual inner product in $\mathbb{R}^d$ and integration is with respect to the Lebesgue measure.
The Fourier series is given by 
\begin{align*}
    \varphi(x) = \sum_{z \in \mathbb{Z}^d} \hat{\varphi}(z) e^{2 \pi \texttt{i} \langle z,x \rangle}.
\end{align*}

%For two integrable functions $F,G : \mathbf{T}^d \to \mathbb{R}$, we define the convolution $F * G (y) = \int_{\mathbf{T}^d} F(x) G(y-x) dx$. We recall that convolution in function space is product in Fourier space.

Recall that we are considering the SGBM with function
\begin{equation}\label{Mfunction2}
    F(x , \sigma_i , \sigma_j) = \begin{cases}
        F_{in}(x), \quad  \text{if } \sigma_i = \sigma_j \\
 F_{out}(x), \quad \text{otherwise,}\\
    \end{cases}
\end{equation}
where the functions $F_{in},F_{out} : \mathbf{T}^d \to [0,1]$ are two measurable functions, known as connectivity probability functions. The expected intracommunity and intercommunity edge densities are given by
\begin{equation}\label{def_mu}
 \mu_{in} = \int_{\mathbf{T}^d} F_{in}(x) dx \textrm{ and } \mu_{out} = \int_{\mathbf{T}^d} F_{out}(x) dx,
 \end{equation}
 the first Fourier modes of the functions $F_{in}$ and $F_{out}$.

Let $A_n$ be the adjacency matrix of a graph $G_n$ on $n$ vertices generated by the SGBM. Let $\lambda_1 \geq \lambda_2 \geq \cdots \geq \lambda_n$ be the eigenvalues of $A_n$, and consider the empirical spectral measure of the matrix $A_n/n$, given by
\begin{equation}\label{def_mu_n}
    \mu_{n} = \sum_{j=1}^n \delta_{\lambda_j/n}.
\end{equation}
Here, $\delta_x$ denotes the Dirac delta measure of $x$. Our first result shows that with high probability $(\mu_n)$ converges in the weak topology  to a counting measure $\mu$ on $\mathbb{R}\setminus (-\xi,\xi)$ for every $\xi>0$. It is a generalization of \cite[Theorem~1]{bordenave2008} and \cite[Theorem 1]{avrachenkov2022}, and corresponds to the first part of the proof of Theorem~\ref{thm_main_informal} described in the introduction. 
\begin{theorem} \label{limitingdensity}
    Consider the SGBM defined by equations \eqref{MAdistribution} and \eqref{Mfunction}.
    Assume that $F_{in}(0)$ and $F_{out}(0)$ are respectively equal to the Fourier series of $F_{in}(\cdot)$ and $F_{out}(\cdot)$ evaluated at 0. Consider the measure
    \begin{align}\label{def:mu}
        \mu = \sum_{z \in \mathbb{Z}^d}\delta_{\frac{\hat{F}_{in}(z) + (k-1) \hat{F}_{out}(z)}{k}} + (k-1) \delta_{\frac{\hat{F}_{in}(z) - \hat{F}_{out}(z)}{k}}. 
    \end{align} For all Borel sets $\mathcal{B}$ with $\mu(\partial \mathcal{B}) = 0$ and $0 \notin \bar{\mathcal{B}}$, the following holds almost surely:
    \begin{align*}
       \lim_{n \to \infty} \mu_n (\mathcal{B}) = \mu (\mathcal{B}).
    \end{align*}
\end{theorem}

Note that, since $\lim_{||z||\rightarrow \infty} \hat{F}_{in}(z) = \lim_{||z||\rightarrow \infty} \hat{F}_{out}(z) = 0 $, the measure $\mu$ defined in~\eqref{def:mu} is indeed a bounded measure if its domain does not contain 0 as an accumulation point.

An application of the Theorem~\ref{limitingdensity} leads to the theorem below, which defines an interval where the eigenvalues of $A$ that are important for the clustering algorithm need to be chosen. It also implies that the other eigenvalues of $A$ are relatively far from this interval.

\begin{theorem}\label{intervalmultiplicity}
    Consider the hypotheses of Theorem~\ref{limitingdensity}, and further assume that $\mu_{in}>\mu_{out}>0$ and    \begin{align}
         \hat{F}_{in}(z) + (k-1)\hat{F}_{out}(z) &\neq {\mu_{in}-\mu_{out}} \quad \forall z \in \mathbb{Z}^d,\label{eq:24} \\
         \hat{F}_{in}(z) - \hat{F}_{out}(z) &\neq {\mu_{in}-\mu_{out}} \quad \forall z \in \mathbb{Z}^d \setminus \{\mathbf{0}\}.\label{eq:25}
    \end{align}
    % {\color{red} Then there exists $\epsilon>0$}
  There exists $\epsilon>0$ such that, for every $\tau$ satisfying $0<\tau<\epsilon$, the following holds a.a.s. There are $k-1$ eigenvalues of $A$ in the interval $I = (\lambda_*-\tau n,\lambda_*+\tau n)$, where $\lambda_* = \frac{n (\mu_{in} - \mu_{out} )}{k}$. Moreover, the distance between $\lambda_*$ and the next nearest eigenvalue of $A$ is at least $n\epsilon$.\color{black}
\end{theorem}

With this, we are ready for the formal statement of our main result, which had been stated informally as Theorem~\ref{thm_main_informal}.
\begin{theorem}\label{thm_main_formal}
Assume that $F$, and $\mu_{in}>\mu_{out}> 0$ are such that the following hold:
\begin{itemize}
\item[(i)] $F_{in}(\mathbf{0})$ and $F_{out}(\mathbf{0})$ are respectively equal to $\hat{F}_{in}(\mathbf{0})$ and $\hat{F}_{out}(\mathbf{0})$.

\item[(ii)] $\displaystyle{\hat{F}_{in}(z) + (k-1)\hat{F}_{out}(z) \neq {\mu_{in}-\mu_{out}} \quad \forall z \in \mathbb{Z}^d}$.

\item[(iii)] $\displaystyle{\hat{F}_{in}(z) - \hat{F}_{out}(z) \neq {\mu_{in}-\mu_{out}} \quad \forall z \in \mathbb{Z}^d \setminus \{0\}}$.
\end{itemize}
Let $k \geq 2$ be fixed. Let $n$ be divisible by $k$ and let $\sigma$ be an assignment of $k$ communities of size $n/k$. And let $X_1,\ldots,X_n$ be u.a.r.\ in $\mathbf{T}^d$.  If $A$ is the adjacency matrix of a graph $G$ generated according to the SGBM with~\eqref{MAdistribution} and~\eqref{Mfunction}, then there is $\epsilon>0$ such that, for any $\tau\in (0,\epsilon)$, the following hold a.a.s.:
\begin{itemize}
\item[(a)] For $\lambda_*=n(\mu_{in}-\mu_{out})/k$, there are $k-1$ eigenvalues $\lambda'_1,\ldots,\lambda'_{k-1}$ of $A$ (including multiplicity) such that $|\lambda_j'-\lambda_*|\leq \tau n$. For any other eigenvalue $\lambda$ of $A$, we have $|\lambda-\lambda_*|\geq \epsilon n$.

\item[(b)] Consider the $n \times (k-1)$ matrix $V$ whose columns are unit eigenvectors of $A$ associated with the eigenvalues $\lambda'_1,\ldots,\lambda'_{k-1}$ of part (a). Consider the embedding of the set $D$ into $\mathbb{R}^{k-1}$ that associates each vertex $i$ with the $i$-th row of $V$. An application of $k$-means clustering to these points produces an estimator $\hat{\sigma}$ such that $\ell(\sigma,\hat{\sigma})\leq \tau \log{n}/n$.

\item[(c)] Assume that $\sigma'$ is the perturbation of $\hat{\sigma}$ obtained as follows: for each $i$, $\sigma'_i=m$ if most neighbors $j$ of $i$ in $G$ satisfy $\hat{\sigma}_j=m$. Then $\ell(\sigma,\sigma')=0$.

\end{itemize}
\end{theorem}

The conditions (i), (ii) and (iii) in the statement of Theorem~\ref{thm_main_formal} are also the technical conditions of Theorem~\ref{thm_main_avrachenkov}, which deals with the case $k=2$. 

An obvious question is whether Theorem~\ref{thm_main_formal} can be applied to natural random graph models.
First consider the SBM where any two points lying in the same cluster are connected with probability $p_{in}$, and any two points in different clusters are connected with probability $p_{out}$, so that $\mu_{in}=p_{in}$ and $\mu_{out}=p_{out}$. Conditions (ii) and (iii) are verified for any choice of $0 \leq p_{in},p_{out} \leq 1$ such that $p_{in} \neq p_{out}$ and $p_{out} \neq 0$. Indeed, 
by Lemma~\ref{append_sbm} we know that $$\hat{F}_{in}(z) = p_{in}\prod_{j=1}^d \text{sinc}(\pi z_j) \text{ and }\hat{F}_{out}(z) = p_{out}\prod_{j=1}^d \text{sinc}(\pi z_j).$$ 
This implies that $F_{in}(\mathbf{0})=p_{in}=\hat{F}_{in}(\mathbf{0})$, $F_{out}(\mathbf{0})=p_{out}=\hat{F}_{in}(\mathbf{0})$, and $\hat{F}_{in}(z)+(k-1)\hat{F}_{out}(z) = \hat{F}_{in}(z)-\hat{F}_{out}(z)= 0$, unless $z_j=0$ for all $j$. When $z_j=0$ for all $j$, we have $\hat{F}_{in}(\mathbf{0})+(k-1)\hat{F}_{out}(\mathbf{0}) = p_{in}+ (k-1)p_{out}$. So, equations~\eqref{eq:24} and~\eqref{eq:25} are always valid for the SBM provided that $p_{in} \neq p_{out}$ and $p_{out} \neq 0$. As a consequence of our results in Section~\ref{sec:daviskahan}, the eigenvalues of an adjacency matrix generated according to the SBM are a.a.s.\ close to the eigenvalues of a block matrix with spectrum $\lambda_1 = \frac{n(p_{in}+(k-1)p_{out})}{k}$, $\lambda_2 =\cdots=\lambda_k =\frac{n(p_{in}-p_{out})}{k}=\lambda_*$ and $\lambda_i=0$ for $i>k$. In this case, the eigenvectors selected by Algorithm~\ref{Cluster_Alg} are associated with $\lambda_2,\ldots,\lambda_{k}$, so that Algorithm~\ref{Cluster_Alg} is the classical spectral clustering algorithm in this case. Recall that the authors of~\cite{lei2015consistency} have shown that using the eigenvectors associated with the largest eigenvalues produces a consistent clustering algorithm for the SBM, even in sparse cases.

Regarding the GBM in the case $k=2$, Proposition 2 in~\cite{avrachenkov2022} established that conditions (i), (ii) and (iii) are almost always verified, i.e., the set of pairs $(r_{in},r_{out})$ such that at least one of the conditions fails has Lebesgue measure 0 in $[0,1]^2$. This may be easily adapted to the case $k\geq 3$, see Lemma~\ref{new_condition_gbm}.

The main new tool for proving Theorem~\ref{thm_main_formal} is Theorem~\ref{error_rate} below. Its proof is a combination of the Davis-Kahan Theorem and some auxiliary Linear Algebra results.  In the statement, we refer to the $n \times n$ matrix $B_{\sigma}=(b_{ij})$ defined as
\begin{equation}\label{def:Bsigma}
b_{ij}=\begin{cases} 
\mu_{in}, & \textrm{ if }\sigma(i)=\sigma(j),\\
\mu_{out}, & \textrm{ if }\sigma(i)\neq\sigma(j).
\end{cases}
\end{equation}
It is easy to prove (see Lemma~\ref{spec_EA}) that $\lambda_* = \frac{\mu_{\text{in}}-\mu_{\text{out}}}{k}n$ is an eigenvalue of $B_\sigma$ with multiplicity $k-1$. For the remainder of this paper, let $\mathcal{U}_{\ell}$ denote the set of all real unitary matrices of order $\ell$, i.e., the set of matrices $Q \in \mathbb{R}^{k \times k}$ such that $QQ^T=Q^TQ=\mathbf{I}_k$. 
\begin{theorem}\label{error_rate}
Consider a $d$-dimensional SGBM satisfying conditions \eqref{MAdistribution} and \eqref{Mfunction} with connectivity probability functions $F_{in}$ and $F_{out}$. Let $G$ be a graph drawn from this SGBM. Let $A$ be the adjacency matrix of $G$ and let $B_\sigma$ defined in~\eqref{def:Bsigma}. 
Let $U=[u_1\cdots u_{k-1}]\in\mathbb{R}^{n\times(k-1)}$, where
$u_1,\ldots,u_{k-1}$ are orthogonal unit eigenvectors of $B_\sigma$ associated with $\lambda_* = \frac{\mu_{\text{in}}-\mu_{\text{out}}}{k}n$, and let $V=[v_1\cdots v_{k-1}]\in\mathbb{R}^{n\times(k-1)}$, where $v_1,\ldots,v_{k-1}$ are the eigenvectors of $A$ associated with the eigenvalues $\lambda'_1,\ldots,\lambda'_{k-1}$ of $A$ closest to $\lambda_ * $. For some $\epsilon>0$, the following holds a.a.s.:
\begin{align*}
    \min \limits_{Q \in\mathcal{U}_{\ell}}\norm{VQ - U}_F   \leq \frac{\sqrt{12 k^5 \log n}}{\epsilon \sqrt{n}}. 
\end{align*}
\end{theorem}
%\frac{\sqrt{(k-1)\frac{4}{k}   (\mu_{in} + (k-1) \mu_{out})  \log n}}{ \Tilde{\epsilon} \sqrt{n}}.

\section{The limiting spectrum of the SGBM}\label{sec:limit_spec}

The aim of this section is to perform the first and the second steps of the proof of Theorem~\ref{thm_main_informal} described in the introduction. Formally, we prove Theorems~\ref{limitingdensity} and~\ref{intervalmultiplicity}.  
\begin{proof}[Proof of Theorem~\ref{limitingdensity}]
  This proof follows the general strategy developed in~\cite[Theorems 1 and 2]{bordenave2008}, which has been extended in~\cite[Theorem 1]{avrachenkov2022} for the two-community block model. 
 We shall use the following notation. Given a measure $\nu$ on the real line and a function $f:\mathbb{R}\rightarrow \mathbb{R}$, we write $\nu(f)=\int_{t\in \mathbb{R}} f(t)~d\nu$. In particular, if $\nu=\nu_n=\sum_{i=1}^n \delta_{\lambda_i}$, we have
\begin{eqnarray}
       \nu(f)=\int_{t\in \mathbb{R}} f(t)~d{\nu_n}&=&\int_{t\in \mathbb{R}} f(t)~d{\delta_{\lambda_1}}+\cdots+\int_{t\in \mathbb{R}} f(t)~d{\delta_{\lambda_n}}\nonumber\\
       &=&f(\lambda_1)+\cdots+f(\lambda_n).\label{measure}
\end{eqnarray} 
  
Let us consider the measure
    \begin{align}\label{mu}
        \mu = \sum_{z \in \mathbb{Z}^d}\delta_{\frac{\hat{F}_{in}(z) + (k-1) \hat{F}_{out}(z)}{k}} + (k-1) \delta_{\frac{\hat{F}_{in}(z) - \hat{F}_{out}(z)}{k}}. 
    \end{align} 
We wish to prove that
    \begin{align}\label{mu_conv}
       \lim_{n \to \infty} \mu_n (\mathcal{B}) = \mu (\mathcal{B})
    \end{align}
holds almost surely for any Borel set $\mathcal{B}$ with $\mu(\partial \mathcal{B}) = 0$ and $0 \notin \bar{\mathcal{B}}$. This is the weak convergence of measures in a domain that does not contain $0$ as an accumulation point.
    
To this end, we let $P_m(t) = t^m$ and we use the method of moments (see Bai and Silverstein \cite[Appendix B]{bai2010spectral}). As the first step, we show that 
  \begin{equation}\label{expectedmeasure}
  \lim \limits_{n \to \infty} \mathbb{E}( \mu_{n}(P_m)) = \mu(P_m).\end{equation}
The second step is an application of Talagrand's inequality to prove that $\mu_{n}(P_m)$ is not far from its mean. Then~\eqref{mu_conv} will follow by applying the Borel-Cantelli Lemma. 

We move to the first step. Let $A$ be the adjacency matrix of a graph $G$. A basic fact in spectral graph theory is that the $(i,j)$ entry of $A^k$ is equal to the number of walks of length $k$ in the graph connecting $i$ to $j$. We have 
\begin{align*}
    \mu_{n}(P_m) &= \frac{1}{n^m} \sum_{i=1}^n \lambda_i^m 
     = \frac{1}{n^m} \tr{A^m}
     = \frac{1}{n^m}\sum_{ \alpha \in [n]^m}  \prod_{l=1}^{m} A(i_l,i_{l+1}),
\end{align*}
    where $\alpha = (i_1,i_2, \ldots , i_m) $ satisfies $i_j\in[n]$ and $i_{m+1} = i_1$ and $A(i_l,i_{l+1})$ denotes the entry $(i_l,i_{l+1})$ of $A$. Note that, in our model, $\mu_n(P_m)$ may be viewed as a random variable that depends on the embedding $X$, as the distribution of $A$ is determined by $X$. Let $\mathcal{A}_n^m$ be the set of such vectors $\alpha=(i_1,\ldots,i_m)$ for which $|\{i_1,\ldots,i_m\}|=m$. This set $\mathcal{A}_n^m$ is known as the set of circular 
    permutations of size $m$. We write 
\begin{align}\label{error_mu}
      \mu_{n}(P_m)  = \frac{1}{n^m}\left[\sum_{ \alpha \in \mathcal{A}_n^m}  \prod_{l=1}^{m} A(i_l,i_{l+1}) + R_m\right].
\end{align}

We first show that the contribution $R_m$ is negligible. Since $A(i,j) \leq  1$ and $\frac{n!}{(n-m)!}=n^m-n^{m-1}\sum\limits_{i=0}^{m-1}i + o(n^{m-1})$, we have 

\begin{align}\label{error_bound}
    R_m &\leq |[n]^m \setminus \mathcal{A}_n^m | 
      = n^m - \frac{n!}{(n-m)!} 
     \leq \frac{m(m-1)n^{m-1}}{2} + o(n^{m-1}).
\end{align}
Thus, $\lim \limits_{n \to \infty} \frac{R_m}{n^m} \to 0$.

Now consider 
\begin{eqnarray}\label{eq_exp}
    \mathbb{E} \left( \sum_{ \alpha \in \mathcal{A}_n^m}  \prod_{l=1}^{m} A(i_l,i_{l+1}) \right) &=&  \sum_{ \alpha \in \mathcal{A}_n^m} \int_{\mathbf{T}^d\times\cdots\times\mathbf{T}^d} \prod_{l=1}^{m} F(x_{i_l}-x_{i_{l+1}}, \sigma_{i_l}, \sigma_{i_{l+1}}) dx_{i_1} dx_{i_2} \cdots dx_{i_m} \nonumber \\
    &=&\sum_{\alpha\in\mathcal{A}_n^m} G(\alpha),
\end{eqnarray}
where $G(\alpha)=\int_{(\mathbf{T}^d)^m} \prod\limits_{l=1}^{m} F(x_{i_l}-x_{i_{l+1}}, \sigma_{i_l}, \sigma_{i_{l+1}}) dx_{i_1} dx_{i_2} \cdots dx_{i_m}$.

Observe that 
$$\prod_{l=1}^{m} F(x_{i_l}-x_{i_{l+1}}, \sigma_{i_l}, \sigma_{i_{l+1}})\stackrel{\eqref{Mfunction}}=\prod_{l\in S(\alpha)} F_{in}(x_{i_l}-x_{i_{l+1}})\prod_{l\in [m]\setminus S(\alpha)}F_{out}(x_{i_l}-x_{i_{l+1}}),$$
where $S(\alpha)=\{j\in[m]:\sigma_{i_j}=\sigma_{i_{j+1}}\}$. Since the integral defining $G(\alpha)$ is over $\mathbf{T}^d$, it depends only on $S(\alpha)$, as we shall see.

\begin{lemat}\cite[Lemma 2]{avrachenkov2022}\label{conv}
     Let $m\in\mathbb{N}$ and $F_1,\ldots,F_m$ be integrable functions over $\mathbf{T}^d$. Then, $$F_1*\cdots*F_m(\mathbf{0})=\int_{(\mathbf{T}^d)^m}\prod_{j=1}^mF_j(x_j-x_{j+1})d_{x_1}\ldots d_{x_m},$$
     with the notation $x_{m+1}=x_1$.
\end{lemat}
Using Lemma~\ref{conv} and the fact that the convolution is commutative, we have
$$G(\alpha)=F_{\text{in}}^{*|S(\alpha)|}*F_{\text{out}}^{*(m-|S(\alpha)|)}(\mathbf{0}).$$
%where $S(\alpha) = \{j \in [m] : \sigma_{i_j} =  \sigma_{i_{j+1}} \}$.

%Moreover, we can see that for a given $p\in\{0,\ldots,m\}$, $G(\alpha)=G(\alpha')=F_{\text{in}}^{*p}*F_{\text{out}}^{*(m-p)}(0)$ for every $\alpha,\alpha'$ such that $|S(\alpha)|=p$. 
Thus, we have
\begin{equation}\label{eq_exp2}
\sum_{\alpha\in\mathcal{A}_n^m} G(\alpha) = \sum_{p=0}^m \sum_{\substack{\alpha\in\mathcal{A}_n^m \\ |S(\alpha)|=p}} F_{\text{in}}^{*p}*F_{\text{out}}^{*(m-p)}(\mathbf{0}).\end{equation}

Since the above expression depends on $p$, but not on the particular choice of $\alpha$, we focus on calculating $|\{\alpha\in\mathcal{A}_n^m : |S(\alpha)|=p\}|$. 
%And, $|S(\alpha)|$ counts the number pairs of consecutive nodes where both lie in the same cluster.
%We consider this calculation in two parts: first, we look at sequences $(\alpha^\ast_1,\ldots,\alpha^\ast_n)$ where each $\alpha^\ast_i$ denotes the cluster in which a vertex is selected. Then, we look at the number of closed paths in the graph associated with this sequence.
Let $\alpha^\ast$ be a vector in $[k]^m$, where we understand $\alpha^\ast_i$ to denote the cluster that contains the $i$-th vertex on the closed walk. Given $\alpha^\ast \in [k]^m$, let $S^\ast(\alpha^\ast)=\{i\in[m]:\alpha^\ast_i=\alpha^\ast_{i+1}\}$. By Theorem~\ref{combinat}, the number of $\alpha^\ast\in[k]^m$ such that $|S(\alpha^\ast)|=p$ is equal to $\binom{m}{p}((k-1)^p + (k-1))$ if $p$ is even and is equal to $\binom{m}{p}((k-1)^p - (k-1))$ if $p$ is odd.

To compute $|\{\alpha\in\mathcal{A}_n^m : |S(\alpha)|=p\}|$, for each $\alpha^\ast\in[k]^m$ such that $|S(\alpha^\ast)|=p$, we compute the number of vectors $\alpha\in\mathcal{A}_n^m$ such that $\alpha_j$ lies in cluster $\alpha^\ast_j$ for all $j$. If $N_i(\alpha^\ast)$ denotes the number of occurrences of $i$ in $\alpha^\ast$, this number is
\begin{equation*}
    \prod_{i=1}^k t_i \text{, where } t_i=\begin{cases}
        \frac{n}{k}\left(\frac{n}{k}-1 \right) \cdots \left(\frac{n}{k}- N_i(\alpha^\ast)+1 \right), \quad  \text{if } N_i(\alpha^\ast) > 0, \\
 1, \quad \text{otherwise.}\\
    \end{cases}
\end{equation*}
It follows that
\begin{equation}\label{eq_exp3}
| \{\alpha\in\mathcal{A}_n^m: |S(\alpha)| = p \}|  = \frac{n^m}{k^m}\binom{m}{p}\left((k-1)^p  + (k-1)(-1)^p \right) + O(n^{m-1}).
\end{equation}

With~\eqref{eq_exp2} and~\eqref{eq_exp3}, equation~\eqref{eq_exp} leads to the following for $\mathbb{E} \left( \sum_{ \alpha \in \mathcal{A}_n^m}  \prod_{l=1}^{m} A(i_l,i_{l+1}) \right) $:
\begin{align*}
    &\sum \limits _{p=0}^m \frac{n^m}{k^m}\binom{m}{p}\left((k-1)^p  + (k-1)(-1)^p \right) {F}_{in}^{*(m-p)}{F}_{out}^{* p}(\mathbf{0}) + O(n^{m-1})\\
    &= \frac{n^m}{k^m}\sum \limits _{p=0}^m\left[\binom{m}{p}(k-1)^pF_{in}^{*(m-p)}F_{out}^{*p}(\mathbf{0}) + (k-1)\binom{m}{p}(-1)^pF_{in}^{*(m-p)}F_{out}^{*p}(\mathbf{0})\right] + O(n^{m-1})\\
    &= n^m\left[\left(\frac{F_{in}+(k-1)F_{out}}{k}\right)^{\ast m}(\mathbf{0}) +(k-1)\left(\frac{F_{in}-F_{out}}{k}\right)^{\ast m}(\mathbf{0})\right] + O(n^{m-1}).
\end{align*}

Now, on the one hand, since $F_{in}(\cdot)$, $F_{out}(\cdot)$ are equal to their Fourier series at $\mathbf{0}$, and $\widehat{F\ast G}(z) = \widehat{F} (z)\widehat{G}(z)$, we have
\begin{align*}
    &\left(\frac{F_{in}+(k-1)F_{out}}{k}\right)^{\ast m}(\mathbf{0}) +(k-1)\left(\frac{F_{in}-F_{out}}{k}\right)^{\ast m}(\mathbf{0})\\
    &= \frac{1}{k^m}\sum_{j=0}^m{m\choose j}\left(F_{in}^{\ast j}(k-1)^{m-j}F_{out}^{\ast m-j}\right)(\mathbf{0}) +(k-1)\frac{1}{k^m}\sum_{j=0}^m{m\choose j}\left(F_{in}^{\ast j}(-1)^{m-j}F_{out}^{\ast m-j}\right)(\mathbf{0})\\
    &= \frac{1}{k^m}\sum_{j=0}^m{m\choose j}(k-1)^{m-j}\sum_{z\in\mathbb{Z}^d}\left(\widehat{F_{in}^{j}F_{out}^{m-j}}\right)(z) +(k-1)\frac{1}{k^m}\sum_{j=0}^m{m\choose j}(-1)^{m-j}\sum_{z\in\mathbb{Z}^d}\left(\widehat{F_{in}^{\ast j}F_{out}^{\ast m-j}}\right)(z)\\   
    &= \sum_{z\in\mathbb{Z}^d}\left[\frac{1}{k^m}\sum_{j=0}^m{m\choose j}(k-1)^{m-j}\left(\widehat{F}_{in}^{j}\widehat{F}_{out}^{m-j}\right)(z) +(k-1)\frac{1}{k^m}\sum_{j=0}^m{m\choose j}\left(\widehat{F}_{in}^{j}\widehat{F}_{out}^{m-j}\right)(z)\right]\\   
    &= \sum_{z \in \mathbb{Z}^d} \left[\frac{1}{k^m}\left(\widehat{F}_{in}+(k-1)\widehat{F}_{out}\right)^{m}(z) +(k-1)\frac{1}{k^m}\left(\widehat{F}_{in}-\widehat{F}_{out}\right)^{ m}(z)\right].   
\end{align*}
On the other hand, by~\eqref{measure} and~\eqref{mu}, we get
\begin{equation*}
    \mu(P_m) = \sum_{z \in \mathbb{Z}^d} \left[\left(\frac{\widehat{F}_{in}+(k-1)\widehat{F}_{out}}{k}\right)^{m}(z) +(k-1)\left(\frac{\widehat{F}_{in}-\widehat{F}_{out}}{k}\right)^{ m}(z)\right] 
\end{equation*}

Combining the above, we obtain 
\begin{align*}
     \mathbb{E}(\mu_{n}(P_m)) &=  \frac{1}{n^m}\left[\mathbb{E} \left( \sum_{ \alpha \in \mathcal{A}_n^m}  \prod_{l=1}^{m} A(i_l,i_{l+1}) \right) + R_m\right] \\
     &=\sum_{z \in \mathbb{Z}^d} \left[\left(\frac{\widehat{F}_{in}+(k-1)\widehat{F}_{out}}{k}\right)^{m}(z) +(k-1)\left(\frac{\widehat{F}_{in}-\widehat{F}_{out}}{k}\right)^{ m}(z)\right]  + o(1) = \mu(P_m) + o(1).
\end{align*}
This concludes the first step.

Moving to the second step, we show that, given $\epsilon>0$
\begin{equation}\label{converge_in_prob}
   \lim_{n\to \infty}\mathbb{P}(|\mu_n(P_m)-\mathbb{E}(\mu_n(P_m))|>\epsilon) = 0.
\end{equation}
Combining with the first step, 
% $\lim\limits_{n\to \infty}\mathbb{E}({\mu_n(P_m)})\to \mu(P_m)$
this establishes that $\mu_n(P_m)$ converges in probability to $\mu(P_m)$.

To show \eqref{converge_in_prob}, we first state some notation that will be useful. By definition, the quantity $\mu_n(P_m)$ is a random variable that depends on the selection of $X$ and of $A$. In this proof, we will often refer to the random selection of $A$ after the set of points $X$ has been fixed, in which case the random variable and its expected value will be denoted by $\mu_n(P_m|X)$ and  $\mathbb{E}\mu_n(P_m|X)$, respectively.
% {\color{red}Let $\mu_n = \mu_n(A, X)$ denote the value obtained from $\mu_n(P_m)$ for a given pair $(X, A)$, where $X\in(\mathbf{T}^d)^n$ and $A$ is the adjacency matrix drawn with probability given by~\eqref{MAdistribution}.} 
% Let $\mu_n=\mu_n(A,X) = \mu_n(P_m)(A,X)$ is the the value obtained in $\mu_n(P_m)$ for a given $X$ and $A$. 
% And, $\mathbb{E}\mu_n(A|X)$ is the expected value for $\mu_n(P_m)$ given a specific $X\in(\mathbf{T}^d)^n$. 
The next two statements will result in \eqref{converge_in_prob}.

\begin{state}\label{talag_1}
    Let $\epsilon'>0$. Given $\epsilon>0$, there exists $n_0$ such that for every $n>n_0$ and $X\in (\mathbf{T}^d)^n$ we have that $\mathbb{P}(|\mu_n(P_m|X)-\mathbb{E}\mu_n(P_m|X)|>\epsilon')<\epsilon$.
\end{state}

\begin{state}\label{talag_2}
    Let $\epsilon>0$. Then, $\lim\limits_{n\to \infty}\mathbb{P}_X(|\mathbb{E}\mu_n(P_m|X)-\mathbb{E}\mu_n(P_m)|>\epsilon)=0$.
\end{state}

Then, combining Statements~\ref{talag_1} and~\ref{talag_2} with the following inequalities for any given $\epsilon>0$, \eqref{converge_in_prob} will follows. Let $\epsilon'>0$. We define $B_{\epsilon'} = \{X:|\mathbb{E}\mu_n(P_m|X)-\mathbb{E}\mu_n(P_m)|<{\epsilon'}/2\}$. Then, given $\epsilon>0$, let $n_0$ be such that for every $n>n_0$, $\mathbb{P}_X(|\mathbb{E}\mu_n(P_m|X)-\mathbb{E}\mu_n(P_m)|>\epsilon'/2)<\epsilon/2$ and 

\begin{align}
    \mathbb{P}(|\mu_n(P_m)-\mathbb{E}\mu_n(P_m)|>\epsilon')&=\mathbb{P}(X\in B_{\epsilon'})\mathbb{P}(|\mu_n(P_m)-\mathbb{E}\mu_n(P_m)|>\epsilon'|X\in B_{\epsilon'})\label{last_tala_1}\\
    &+\mathbb{P}(X\in (\mathbf{T}^d)^n\setminus B_{\epsilon'})\mathbb{P}(|\mu_n(P_m)-\mathbb{E}\mu_n(P_m)|>\epsilon'|X\in (\mathbf{T}^d)^n\setminus B_{\epsilon'})\label{last_tala_2}.
\end{align}

Now, \begin{align*}
RHS \ of \ \eqref{last_tala_1} &\leq \int_{B_{\epsilon'}}\mathbb{P}(X=(x_1,\ldots,x_n))\mathbb{P}(|\mu_n(P_m|X)-\mathbb{E}\mu_n(P_m)|>\epsilon')d_{x_1}\ldots d_{x_n}\\
&\stackrel{\eqref{talag_2}}\leq \int_{B_{\epsilon'}}\mathbb{P}(X=(x_1,\ldots,x_n))\mathbb{P}(|\mu_n(P_m|X)-\mathbb{E}\mu_n(P_m)|>{\epsilon'}/2)d_{x_1}\ldots d_{x_n}\\
&\stackrel{\eqref{talag_1}}\leq\int_{B_{\epsilon'}}\epsilon/2dX\leq\epsilon/2.\\
\end{align*}
And, 
\begin{align*}
RHS \ of \ \eqref{last_tala_2} &\stackrel{\eqref{talag_2}}\leq \epsilon/2\mathbb{P}(|\mu_n(P_m)-\mathbb{E}\mu_n(P_m)|>\epsilon'|X\in (\mathbf{T}^d)^n\setminus B_{\epsilon'})\leq \epsilon/2\\
\end{align*}

We note that if the function $F(X_i-X_j,\sigma_i,\sigma_j)$ is deterministic (e.g., in the GBM where the values achieved by $F$ are always $0$ or $1$), the proof can be done in one step, proving only Statements~\ref{talag_2}. 
For the general case, it remains to prove the Statement~\ref{talag_1} and Statement~\ref{talag_2}.
                                                                                                                                                                                                                                                                                                                                                                                                                                                                                                                                                                                                                                                                        
\begin{proof}[Proof of Statement~\ref{talag_1}]
Let $X=\{x_1,\ldots,x_n\}\subset \mathbf{T}^d$ be fixed. The distribution of the adjacency matrix is determined by these points, as the entry $a_{ij}=a_{ji}$ is equal to $1$ with probability $F(x_i-x_j,\sigma_i,\sigma_j)$ and $0$ with probability $1-F(x_i-x_j,\sigma_i,\sigma_j)$. Now, consider the map 
\begin{align*}
Q_m^X : \{0,1\}^{{n\choose 2}} &\longrightarrow \mathbb{R} \\
       A &\longmapsto \frac{1}{n^{m-1}} \tr A^m.
\end{align*}

We now state a result that shows that $Q_m^X$ is Lipschitz.
\begin{lemat}~\cite[Lemma 5]{avrachenkov2022}
    Let $A,\Tilde{A}\in\{0,1\}^{n\times n}$ be two adjacent matrices, and $m\geq 1$. Then,
    $$\left|\tr(A^m)-\tr(\Tilde{A}^m) \right|\leq m n^{m-2}d_H(A,\Tilde{A}).$$
\end{lemat}

Let $M_m$ be the median of $Q_m^X$. Then, by Talagrand's inequality~\cite[Proposition 2.1]{talagrand1996new}, we have that

\begin{align*}
    \mathbb{P}(|Q_m^X(A) - M_m| > t) \leq 4\exp\left({-\frac{(\frac{t}{m/n})^2}{{n\choose 2}}}\right)\leq4\exp\left({-\frac{t^2}{m^2}}\right),
\end{align*}
where the probability space was a product of ${n \choose 2}$ probability spaces.

Further, since $|Q_m^X(A)-M_m|$ is a positive random variable,
\begin{align*}
    \mathbb{E}(|Q_m^X(A)-M_m|) &= \int_t \mathbb{P}(|Q_m^X(A) -M_m| > t)  dt \\
    & \leq \int_t 4e^{-\frac{t^2}{m^2}} dt =: C_m .
\end{align*}
Next, consider 
\begin{align*}
    |Q_m^X(A) - \mathbb{E}Q_m^X| &\leq |Q_m^X(A) - M_m| + |M_m - \mathbb{E}Q_m^X | \\
    &\leq |Q_m^X(A) - M_m| + \mathbb{E}|M_m - Q_m^X | \\
    & \leq |Q_m^X(A) - M_m| + C_m.  
\end{align*}
Now, note that $\mathbb{E}Q_m^X = n\mathbb{E}\mu_n(P_m|X)$, which implies that
\begin{align*}
    \mathbb{P}(|\mu_n(P_m) - \mathbb{E}\mu_n(P_m|X)| > s) &= \mathbb{P}\left(\frac{1}{n}|Q_m^X(A) - \mathbb{E}Q_m^X| > s\right) \\
    &= \mathbb{P}(|Q_m^X(A) - \mathbb{E}Q_m^X| > ns) \\
    &\leq \mathbb{P}(|Q_m^X(A) - M_m| + C_m > ns )\\
    &= \mathbb{P}(|Q_m^X(A) - M_m| > ns - C_m).\\
\end{align*}

Again by applying Talagrand's inequality, we obtain

\begin{align*}
    \mathbb{P}(|\mu_n(P_m) - \mathbb{E}\mu_n(P_m|X)| > s) \leq 4\exp\left({- \frac{n^2(ns -C_m)^2 }{m^2{n\choose 2}}}\right)\leq 4\exp\left({- \frac{1}{m^2}}\left(ns - C_m\right)^2 \right).
\end{align*}

Choosing $s_n = \frac{C_m}{n^\kappa}$ for $0<\kappa<1$ and defining $\epsilon_n = 4\exp\left({- \frac{1}{m^2}\left(ns_n - C_m\right)^2 }\right)$, we will have 
% \begin{align*}
%     \epsilon_n = 4\exp\left({- \frac{1}{m^2}\left(ns_n - C_m\right)^2 }\right)
% \end{align*}
\begin{align*}
  \sum_{n=1}^\infty   \mathbb{P}(|\mu_n(P_m) - \mathbb{E}\mu_n(P_m|X)| > s_n)  \leq \sum_{n=1}^\infty \epsilon_n < \infty
 \end{align*}
 
Then using Borel-Cantelli Lemma the statement is proved.
 
\end{proof}

\begin{proof}[Proof of Statement~\ref{talag_2}]

Now, consider the following map $\Tilde{Q}_m:(\mathbf{T}^d)^n \to \mathbb{R}$ given by $\Tilde{Q}_m(X) = \mathbb{E}\left(\frac{\tr(A^m)}{n^{m-1}}|X \right)$. 

Note that, for given $X$, the entries of $A$ are generated with probability $F(x_i-x_j,\sigma_i,\sigma_j)$. Similar to~\eqref{error_mu} we consider the $A^m$ as follows
\begin{align}
      \left[\sum_{ \alpha \in \mathcal{A}_n^m}  \prod_{l=1}^{m} A(i_l,i_{l+1}) + R_m\right].
\end{align}
and as in the inequality~\label{error_bound} $R_m\leq K'_mn^{m-1}+o(n^{m-1})$. Then,

\begin{align*}
    \mathbb{E}\left(\frac{\tr(A^m)}{n^{m-1}}|X\right)
    &=\frac{1}{n^{m-1}}\mathbb{E}\left[\sum_{ \alpha \in \mathcal{A}_n^m}  \prod_{l=1}^{m} A(i_l,i_{l+1}) + R_m\right]\\
    &=\frac{1}{n^{m-1}}\left[\sum_{ \alpha \in \mathcal{A}_n^m}  \mathbb{E}\prod_{l=1}^{m} A(i_l,i_{l+1}) + \mathbb{E}R_m\right]\\
    &=\frac{1}{n^{m-1}}\sum_{ \alpha \in \mathcal{A}_n^m} \prod_{l=1}^{m} F(X_{i_j}-X_{i_{j+1}},\sigma_i,\sigma_j)  + \frac{1}{n^{m-1}}\mathbb{E}R_m
\end{align*}

% So $\mathbb{E}\left(\frac{A^m}{n^{m-1}}|X\right)=\frac{\mathbb{E}(A|X)^m}{n^{m-1}}$. And the entry $(\mathbb{E}(A|X))_{ij}$ will be $F(x_i-x_j,\sigma_i,\sigma_j)$.

First, we show that, for $n$ sufficiently large and for $X,X'\in (\mathbf{T}^d)^n$, it is true that $|\Tilde{Q}_m(X) - \Tilde{Q}_m(X')| \leq 2K_m d_H(X,X')$, where $K_m$ is constant that depends only on $m$ and $d_H(X,X')$ is the hamming distance between $X$ and $X'$, that is, $d_H(X,X')=\left|\{i\in[m]:x_i\neq x_i'\}\right|$. So, we choose $n\geq n_0$ such that $\frac{1}{n^{m-1}}R_m\leq K'_m+\frac{1}{n^{m-1}}o(n^{m-1})\leq 2K'_m$.

Let $X,X'\in (\mathbf{T}^d)^n$ be such that there is $\ell$ positions of $X$ different from $X'$. Then, of course, $d_H(X,X')=\ell$ and $|\Tilde{Q}_m(X) - \Tilde{Q}_m(X')|$ will be less or equal than
\begin{align*}
     &\frac{1}{n^{m-1}}\left[\left|\sum \limits_{i_1,i_2,\cdots ,i_m}  \prod_{j=1}^m F(X_{i_j}-X_{i_{j+1}},\sigma_i,\sigma_j) - \prod_{j=1}^m F(X'_{i_j}-X'_{i_{j+1}},\sigma_i,\sigma_j)\right|+|\mathbb{E}R_m(X)|+|\mathbb{E}R_m(X')|\right]\\
    &\leq \frac{1}{n^{m-1}} \sum \limits_{i_1,i_2,\cdots ,i_m}  \left|\prod_{j=1}^m F(X_{i_j}-X_{i_{j+1}},\sigma_i,\sigma_j) - \prod_{j=1}^m F(X'_{i_j}-X'_{i_{j+1}},\sigma_i,\sigma_j)\right|+K'_m+\frac{1}{n^{m-1}}o(n^{m-1})
\end{align*}

Note that when the indices $i_1,i_2,\cdots ,i_m$ do not contain the changed node, we have the difference term to be zero. When it has changed index, the difference between the product term is at most 1. The number of possibilities of $i_1,i_2,\cdots ,i_m$ contains a changed node is $n^{m-1}m\ell$, since at least one position needs to be one of the $\ell$ changed nodes, while the others can assume any $n$ node. Thus

\begin{align*}
    |Q_m(X) - Q_m(X')| \leq m\ell+2K'_m\leq md_H(X,X')+2K'_md_H(X,X') = K_md_H(X,X'). 
\end{align*}

Now, let $M_m$ be the median of $\Tilde{Q}_m$. Then, again by Talagrand's inequality, we have that

\begin{align*}
    \mathbb{P}(|\Tilde{Q}_m(X) -M_m| > t) \leq 4e^{-\frac{t^2}{4K_m^2n}}.
\end{align*}

Since $|\Tilde{Q}_m(X) -M_m|$ is a positive random variable, we can write
\begin{align*}
    \mathbb{E}(|\Tilde{Q}_m(X) -M_m|) &= \int_t \mathbb{P}(|\Tilde{Q}_m(X) -M_m| > t)  dt \\
    & \leq \int_t 2e^{-\frac{t^2}{4K_m^2n}} dt \\
    &= C_m \sqrt{n}.
\end{align*}
Further consider 
\begin{align*}
    |\Tilde{Q}_m(X) - \mathbb{E}\Tilde{Q}_m| &\leq |\Tilde{Q}_m(X) - M_m| + \mathbb{E}|M_m - \Tilde{Q}_m | \\
    & \leq |\Tilde{Q}_m(X) - M_m| + C_m \sqrt{n}.  
\end{align*}

Now, for the remainder of this proof it is important to note the following
\begin{align*}
\mathbb{E}\Tilde{Q}_m =\mathbb{E}\left(\mathbb{E}\left(\frac{A^m}{n^{m-1}}|X\right)\right)=\mathbb{E}\left(\mathbb{E}\left(n\mu_n(P_m)|X\right)\right)=n\mathbb{E}\mu_n(P_m).
\end{align*}
Of course, besides that, $\Tilde{Q}_m(X) = \mathbb{E}\mu_n(P_m|X)$. Thus,
\begin{align*}
    \mathbb{P}(|\mathbb{E}\mu_n(P_m|X) - \mathbb{E}\mu_n(P_m)| > s) &=\mathbb{P}(\frac{1}{n}|\Tilde{Q}_m(X) - \mathbb{E}\Tilde{Q}_m| > s)\\ &\leq \mathbb{P}(|\Tilde{Q}_m(X) - M_m| > ns - C_m\sqrt{n})\\ 
    &\leq \mathbb{P}\left(|\Tilde{Q}_m(X) - M_m| > n\left(s - \frac{C_m}{\sqrt{n}}\right)\right).
\end{align*}

Again by applying Talagrand's inequality, we obtain
\begin{align*}
    \mathbb{P}(|\mathbb{E}\mu_n(P_m|X) - \mathbb{E}\mu_n(P_m)| > s) \leq 4\exp\left(- \frac{n(s - \frac{C_m}{\sqrt{n}})^2 }{4K_m^2}\right)
\end{align*}

Choosing $s = \frac{C_m}{\sqrt{n}} + \epsilon$, and using Borel-Cantelli Lemma, we achieve the desired result. 

\end{proof}
\end{proof}

\begin{proof}[Proof of the Theorem~\ref{intervalmultiplicity}]

 Because $F_{in}$ and $F_{out}$ are integrable, we have  $\lim \limits_{\norm{z}_{\infty} \to \infty}\hat{F}_{out}(z) = 0$ and  $\lim \limits _{\norm{z}_{\infty} \to \infty} \hat{F}_{in}(z) = 0$ (see~\cite[Proposition 3.2.1]{grafakos2008fourier}). 

We shall prove that there are only $k-1$ eigenvalues of $\frac{A}{n}$ near $\frac{\mu_{in} - \mu_{out}}{k}$ for large $n$.  

Let $\epsilon_0=(\mu_{in}-\mu_{out})/2k$.
%Consider $\epsilon_1$ such that $0 < \epsilon_1 < |\frac{\hat{F}_{in}(z)+ (k-1)\hat{F}_{out}(z)}{k} - \frac{\mu_{in} - \mu_{out}}{k} |$ for all $z \in \mathbb{Z}^d$. 
Given that $\hat{F}_{in}(z)+ (k-1)\hat{F}_{out}(z)$ tends to $0 \neq \mu_{in}-\mu_{out}$ as $\norm{z} \rightarrow \infty$, fix $M$ such that 
$$\left|\frac{\hat{F}_{in}(z)+ (k-1)\hat{F}_{out}(z)}{k} - \frac{\mu_{in} - \mu_{out}}{k} \right| \geq \epsilon_0$$ for all $z \in \mathbb{Z}^d$ such that $\norm{z} \geq M$.

There are only finitely many choices for $z\in\mathbb{Z}^d$ such that $\norm{z} < M$. For these choices of $z$, we have $\hat{F}_{in}(z)+ (k-1)\hat{F}_{out}(z) \neq \mu_{in}-\mu_{out} $ by~\eqref{eq:24}. So we may fix $\epsilon_1$ such that
    $$0<\epsilon_1 < \min_{z\in\mathbb{Z}^d}\left(\left|\frac{\hat{F}_{in}(z)+ (k-1)\hat{F}_{out}(z)}{k} - \frac{\mu_{in} - \mu_{out}}{k}\right|\right).$$
    
For the same reason we can fix $\epsilon_2$ such that $0 < \epsilon_2 < |\frac{\hat{F}_{in}(z)- \hat{F}_{out}(z)}{k} - \frac{\mu_{in} - \mu_{out}}{k} |$ for all $z \neq 0$.  

Let $\epsilon = \min \{ \epsilon_1 , \epsilon_2\}$. Fix $0<\tau<\epsilon$. By Theorem~\ref{limitingdensity},  the intervals $B_1 = ( \frac{\mu_{in} - \mu_{out}}{k}- \tau , \frac{\mu_{in} - \mu_{out}}{k} + \tau) $ and $B_2 = ( \frac{\mu_{in} - \mu_{out}}{k}- \epsilon , \frac{\mu_{in} - \mu_{out}}{k} + \epsilon)$ satisfy $\mu(B_1) = \mu(B_2) = k-1$. As a consequence, a.a.s.\ $k-1$ eigenvalues $\lambda_1'/n, \ldots, \lambda_k'/n$ of $A/n$ satisfy $|\lambda_i'/n-(\mu_{in} - \mu_{out})/k|\leq \tau$ while the remaining eigenvalues $\lambda_j'/n$ satisfy $|\lambda_j'/n-(\mu_{in} - \mu_{out})/k|\geq \epsilon$. This establishes the needed result.
\end{proof}

\section{Proof of Theorem~\ref{error_rate}}\label{sec:daviskahan}

The aim of this section is to prove Theorem~\ref{error_rate}, which relates the eigenvectors of a matrix generated according to the SGBM with the eigenvectors of a much simpler matrix. Although identifying the community assignment $\sigma$ is the objective of our algorithm, in this section there is no loss of generality in assuming that $\sigma$ is the assignment such that $1,2,\ldots,n/k$ lie in the first community, $n/k+1,n/k+2,\ldots,2n/k$ lie in the second community, and so on. Then the matrix $B_\sigma$ defined in~\eqref{def:Bsigma} is just a block matrix with diagonal blocks being constant matrices with entries equal to $\mu_{in}$, while the remaining blocks have entries equal to $\mu_{out}$. 

We start defining a useful operation to study the spectrum of $B_{\sigma}$. Given an $m \times n$ matrix $A$ and a $p\times q$ matrix $B$, their Kronecker product $A\otimes B$ is the $pm\times qn$ matrix:
\[
A \otimes B = \begin{bmatrix}
a_{11}B & a_{12}B & \cdots & a_{1n}B \\
a_{21}B & a_{22}B & \cdots & a_{2n}B \\
\vdots & \vdots & \ddots & \vdots \\
a_{m1}B & a_{m2}B & \cdots & a_{mn}B
\end{bmatrix}.
\]
The property below follows easily from the definition of the Kronecker product.
\begin{propt}\label{tprod}
    If $(\lambda, v)$ is an eigenpair of $A$ and $(\nu ,u )$ is an eigenpair of $B$, then $(\lambda \nu, v \otimes u )$ is an eigenpair of $A \otimes B$. 
\end{propt}

For $k=3$, we have
\begin{align*}
    B_\sigma = \begin{bmatrix}
        \mu_{in} & \mu_{out} & \mu_{out} \\
         \mu_{out} & \mu_{in} &  \mu_{out}  \\
         \mu_{out} &   \mu_{out} & \mu_{in} 
    \end{bmatrix} \otimes \mathbf{J}_{\frac{n}{3}}, 
\end{align*}
where $\mathbf{J}_{\frac{n}{3}}$ is the all $1$ matrix with dimension $\frac{n}{3}\times \frac{n}{3}$. In general we have 
\begin{align}\label{expected-adjacency}
    B_{\sigma} = ((\mu_{in} - \mu_{out}) \mathbf{I}_k + \mu_{out}\mathbf{J}_k) \otimes  \mathbf{J}_{\frac{n}{k}},
\end{align}
where $\mathbf{I}_k$ is the identity matrix of order $k$.

%Consider the SGBM defined by~\eqref{MAdistribution} and \eqref{Mfunction} and $k=3$. 
%The expected adjacency matrix is of the form
%\begin{align*}
%    \mathbb{E}A_3 = \begin{bmatrix}
 %       \mu_{in} & \mu_{out} & \mu_{out} \\
  %       \mu_{out} & \mu_{in} &  \mu_{out}  \\
   %      \mu_{out} &   \mu_{out} & \mu_{in} 
    %\end{bmatrix} \otimes \mathbf{J}_{\frac{n}{3}}-\mu_{in}I_n, 
%\end{align*}
%where $\mathbf{J}_{\frac{n}{3}}$ is an all $1$ matrix with dimension $\frac{n}{3}\times \frac{n}{3}$. And, in general we have 
%\begin{align}\label{expected-adjacency}
    %\mathbb{E} A_k = ((\mu_{in} - \mu_{out}) I_k + \mu_{out}\mathbf{J}_k) \otimes  \mathbf{J}_{\frac{n}{k}}-\mu_{in}I_n.
%\end{align}

\begin{lemat}\label{spec_EA}
The nonzero eigenvalues of $B_\sigma$ are precisely
    \begin{enumerate}
        \item[(i)] $\frac{n}{k}(\mu_{in}+(k-1)\mu_{out})$ with multiplicity one. Its eigenspace is generated by the all ones vector $\mathbf{1}$.
        \item[(ii)] $\lambda_* = \frac{n}{k}(\mu_{in}-\mu_{out})$, with multiplicity $k-1$. Its eigenspace is generated by the columns of the matrix $U$ given by
        \begin{equation}\label{eigenvectorEA}
    U(i,j) = \begin{cases}
        \sqrt{\frac{k}{2n}}, &  \text{if } i\leq \frac{n}{k}, \\
        -\sqrt{\frac{k}{2n}}, &  \text{if } (j+1)\frac{n}{k}<i\leq (j+2)\frac{n}{k},\\
0, & \text{otherwise.}\\
    \end{cases}
\end{equation}
    \end{enumerate}
\end{lemat}

\begin{proof}
    % Property~\ref{tprod} will play an important role in this proof.
    It is easy to check that $\mathbf{1}_{\frac{n}{k}}$ is an eigenvector of $\mathbf{J}_{\frac{n}{k}}$ associated with the eigenvalue $\frac{n}{k}$, since each row of $\mathbf{J}_{\frac{n}{k}}$ adds to $\frac{n}{k}$. Since $\mbox{rank}(\mathbf{J}_{\frac{n}{k}})=1$, the other eigenvalues are $0$.

   Consider $(\mu_{in} - \mu_{out}) \mathbf{I}_k + \mu_{out}\mathbf{J}_k$. The eigenvalues of $\mu_{out}\mathbf{J}_k$ are $k\mu_{out}$, with multiplicity one, and $0$.  A basis for the eigenspace of $0$ is given by the columns $u_1',\ldots,u_{k-1}'$ of
\begin{equation}\label{defUprime}
    U' = \begin{bmatrix}
    u'_1 & \cdots & u'_{k-1} 
\end{bmatrix}= \begin{bmatrix}
        \frac{1}{\sqrt{2}} & \frac{1}{\sqrt{2}} & \ldots & \frac{1}{\sqrt{2}} \\
        -\frac{1}{\sqrt{2}} & 0 & 0 & 0  \\
         0 &   -\frac{1}{\sqrt{2}} & 0 & 0 \\
         \vdots &   \vdots & \vdots & \vdots \\
         0 &   0 & 0 & -\frac{1}{\sqrt{2}} 
    \end{bmatrix}. 
\end{equation}  
   So, the eigenvalues of $(\mu_{in} - \mu_{out}) \mathbf{I}_k + \mu_{out}\mathbf{J}_k$ are $(\mu_{in}-\mu_{out}) + k\mu_{out}=\mu_{in}+(k-1)\mu_{out}$, with eigenspace generated by $\mathbf{1}_k$, and $(\mu_{in}- \mu_{out})+0$ with eigenspace generated by $u'_1,u'_2, \ldots, u'_{k-1}$. 

    By Property~\ref{tprod}, the nonzero eigenvalues of $B_\sigma = ((\mu_{in} - \mu_{out}) \mathbf{I}_k + \mu_{out}\mathbf{J}_k) \otimes  \mathbf{J}_{\frac{n}{k}}$ are
    \begin{enumerate}
    \item [(i)] $\frac{n}{k}(\mu_{in}+(k-1)\mu_{out})$, with multiplicity one and associated eigenvector $\mathbf{1}$.
    \item [(ii)] $\lambda = \frac{n}{k}(\mu_{in}-\mu_{out})$, with multiplicity $k-1$, with orthogonal eigenvectors $u_1=u'_1\otimes \mathbf{1}_{\frac{n}{k}},u_2=u'_2\otimes \mathbf{1}_{\frac{n}{k}}, \ldots, u_{k-1}=u'_{k-1}\otimes \mathbf{1}_{\frac{n}{k}}$.
\end{enumerate}

So, the eigenvectors $u_j$ of $B_\sigma$, $1\leq i \leq k-1$ are the columns of $U = \begin{bmatrix}
    u'_1 & \cdots & u'_{k-1} 
\end{bmatrix} \otimes \sqrt{\frac{k}{n}}\mathbf{1}_{\frac{n}{k}}$, which are precisely the columns of the matrix $U$ in the statament of the lemma.
\end{proof}

Let $\mathbb{E}A$ be the expected adjacency matrix of a graph chosen according to the SGBM satisfying conditions \eqref{MAdistribution} and \eqref{Mfunction}. This means that the probability that two points $i$ and $j$ are connected is 0 if $i=j$, it is $\mu_{in}$ if $i \neq j$ and $\sigma(i)=\sigma(j)$, and it is $\mu_{out}$ if $i \neq j$ and $\sigma(i)\neq\sigma(j)$. As a consequence, we have 
$$\mathbb{E}A=B_\sigma-\mu_{in} \mathbf{I}_n.$$
The eigenvectors of $\mathbb{E}A$ and $B_\sigma$ are the same, and the eigenvalues of $\mathbb{E}A$ are $\frac{n}{k}(\mu_{in}+(k-1)\mu_{out})-\mu_{in}$, $\alpha_*-\mu_{in}$ and $-\mu_{in}$, respectively.
% In \eqref{expected-adjacency}, if we look at the matrix $(\mu_{in} - \mu_{out}) I_k + \mu_{out}\mathbf{J}_k$, it has all one vector $\mathbf{1}$ with eigenvalue $\mu_{in} + (k-1) \mu_{out}$, and $k-1$ vectors, let those eigenvectors be $u_2,u_3 \ldots u_{k}$, which are orthogonal to all one vector, having eigenvalue $\mu_{in} - \mu_{out}$. \par
% On the other hand, the matrix $\mathbf{11}^T_{\frac{n}{k}}$ has eigenvalues $\frac{n}{k}$ with all one vector as the eigenvector and remaining vectors orthogonal to it as the eigenvectors with eigenvalue zero. Since, we want normalized eigenvectors, we will call $\sqrt{\frac{k}{n}}\mathbf{1}_{\frac{n}{k}}$ the eigenvector associated to the eigenvalue $\frac{n}{k}$ which is multiple of the all one eigenvector.\par
% From the tensor product property \ref{tprod}, we have $\frac{n (\mu_{in} - \mu_{out})}{k}$ as an eigenvalue of $\mathbb{E}A_k$ with eigenvectors $u_i \otimes \mathbf{1}_{\frac{n}{k}}$ for $i = 2,3, \cdots ,k$. 

%Thus if we look at the eigenspace of $\mathbb{E}A_k$ corresponding to eigenvalue $ \lambda_* = \frac{n (\mu_{in} - \mu_{out})}{k}-\mu_{in}$, there are $k-1$ distinct points in that space. We shall observe that we will want unitary eigenvectors and that, when $n$ is large, $\lambda_*$ is just $\frac{n}{k} (\mu_{in} - \mu_{out})$.

We shall use the following result about the rows of the matrix $U$ defined in~\eqref{eigenvectorEA}, whose proof is straightforward.
\begin{lemat}\label{distancecentroids}
    Let $i,\ell\in[k]$ with $i\neq \ell$. And, let $w_{i\frac{n}{k}+j_1}$ and $w_{\ell\frac{n}{k}+j_2}$ be the $(i\frac{n}{k}+j_1)$-th and $(\ell\frac{n}{k}+j_2)$-th rows of $U$ for some $j_1,j_2\in[\frac{n}{k}]$. Then,
    \begin{equation*}
        \norm{w_{i\frac{n}{k}+j_1}-w_{\ell\frac{n}{k}+j_2}}_2 \geq \sqrt{\frac{k}{n}}.
    \end{equation*}
\end{lemat}

We shall also use the following version of the Chernoff bound.
\begin{lemat}\cite[Corollary 4.6]{mitzenmacher2017probability}\label{chernoff}
Suppose that $X_1,\ldots,X_n$ are independent random variables taking values in $\{0, 1\}$. Let $X$ denote their sum and consider the expected value $\mu(X) = \mathbb{E}[X]$. Then for any $0<\gamma<1$,
    $$\mathbb{P}(|X-\mu(X)|>\gamma\mu(X))\leq2\exp\left(-\frac{\gamma^2\mu(X)}{3}\right).$$
\end{lemat}

We shall apply a version of the Davis-Kahan Theorem given in \cite[Theorem 3.2]{li1998relative}. Here, for a matrix $M=(m_{ij})$, we use its Frobenius norm $\norm{M}_F=\tr(M^TM)=\left(\sum_{i,j} M(i,j)^2  \right)^{1/2}$. For the results below, the notation $Q=[Q_0,Q_1]$ means that the columns of $Q$ are split into a $k \times d$ matrix $Q_0$ and a $k\times (k-d)$ matrix $Q_1$, for some integer $d$ satisfying $1 \leq d \leq k-1$.
\begin{theorem}[Davis-Kahan]\label{Davis-Kahan}
 Consider symmetric $k \times k$ matrices $M$ and $\Tilde{M}=M+H$. Let
 $M = E_0 \Lambda_0 E_0^T + E_1 \Lambda_1 E_1^T$ and $\tilde{M}= F_0 \Gamma_0 F_0^T + F_1 \Gamma_1 F_1^T$ be the eigendecompositions of $M$ and $\Tilde{M}$, respectively, where $[E_0, E_1]$ and $[F_0, F_1]$ are both unitary  matrices such that $E_0$ and $F_0$ are $k \times d$. Suppose that there is an interval $[a,b]$ and a constant $\epsilon>0$ such that the spectrum of $\Lambda_0$ is contained in $ [a,b]$, while  the diagonal elements of $\Gamma_1$ lie in $\mathbb{R}\setminus(a-\epsilon,b+\epsilon)$.
 Then 
    \begin{align*}
        \norm{F_1^T E_0}_F \leq \frac{\norm{F_1^T H E_0}_F}{\epsilon}.
    \end{align*} 
\end{theorem}

Moreover, we use the fact that $\norm{F_1^T H E_0}_F \leq \norm{F_1}_F \norm{HE_0}_F$ and the fact that each column of $F_1$ is a unit vector to obtain 
\begin{align}
    \norm{F_1^T E_0}_F \leq (k-1) \frac{\norm{H E_0}_F}{\epsilon}. \label{newineq}
\end{align}

We shall write $\norm{F_1^T E_0}_F$ in terms of $E_0$ and $F_0$.
Recall that the trace of the product is invariant under circular shifts, that is, 
\begin{equation}\label{trcircuit}
\tr(ABCD)=\tr(DABC)=\tr(CDAB)=\tr(BCDA)
\end{equation}

\begin{lemat}\label{lemma_43}
       Let $E_0,F_0 \in \mathbb{R}^{n \times d}$ and $E_1,F_1 \in \mathbb{R}^{n \times (n-d)}$ be matrices such that $[E_0, E_1]$ and $[F_0, F_1]$ are both orthogonal matrices. Then, $\norm{F_1^TE_0}_F = \frac{\norm{E_0E_0^T - F_0F_0^T}_F}{\sqrt{2}}$
\end{lemat}

\begin{proof}
The expression $\norm{F_1^TE_0}_F^2$ may be rewritten as
    \begin{eqnarray}   \label{eq:aux} 
    \norm{F_1^TE_0}_F^2 &=& \tr ((F_1^TE_0)^T(F_1^TE_0) )\nonumber\\
    &=& \tr(E_0^TF_1F_1^TE_0) \nonumber\\
    &=& \tr(E_0^T(\mathbf{I}_k-F_0F_0^T)E_0)\\
    &=& \tr(E_0^TE_0) - \tr(E_0^TF_0F_0^TE_0)\nonumber\\
    &=& d - \tr(E_0^TF_0F_0^TE_0).\nonumber
    \end{eqnarray}
Also, when computing $\norm{E_0E_0^T - F_0F_0^T}_F^2$ we have
    \begin{eqnarray*}    
    \norm{E_0E_0^T - F_0F_0^T}_F^2 &=& \tr (E_0E_0^TE_0E_0^T) + \tr(F_0F_0^TF_0F_0^T) - \tr(E_0E_0^TF_0F_0^T)- \tr(F_0F_0^TE_0E_0^T)\\
    &=& \tr (E_0E_0^T) + \tr(F_0F_0^T) - 2\tr(E_0E_0^TF_0F_0^T)\\
    &=& 2d - 2\tr(E_0^TF_0F_0^TE_0)=2\norm{F_1^TE_0}_F^2.
    \end{eqnarray*}
    Then, $\norm{F_1^TE_0}_F = \frac{\norm{E_0E_0^T - F_0F_0^T}_F}{\sqrt{2}}$.
\end{proof}

\begin{lemat}\label{unitary-lemma} 
    Let $E_0,E_1,F_0,F_1$ be matrices such that $E_0,F_0 \in \mathbb{R}^{n \times d}$ and $[E_0, E_1]$ and $[F_0, F_1]$ are orthogonal. Let $S \Sigma P^T$ be the singular value decomposition of $E_0^TF_0$. For $\Tilde{Q} = PS^T$, we have
    $$\norm{F_0\Tilde{Q} - E_0}_F =\inf \limits_{Q \in\mathcal{U}_{d}} \norm{F_0Q - E_0}_F\leq \sqrt{2} \norm{F_1^TE_0}_F.$$ 
\end{lemat}

\begin{proof}
Let $Q\in\mathcal{U}_d$. Expanding $\norm{F_0Q-E_0}_F^2$ we have,
\begin{eqnarray}\label{aux_eq_CH}
    \norm{F_0Q - E_0}_F^2 &=&  \tr \left( (F_0Q - E_0)^T(F_0Q - E_0)  \right) \nonumber\\
    &\stackrel{(*)}=&  \tr(Q^TF_0^TF_0Q) -\tr(E_0^TF_0Q) - \tr(Q^TF_0^TE_0) + \tr(E_0^TE_0)  \nonumber \\
    &\stackrel{(**)}=& d + d - \tr(Q^TF_0^T E_0) - \tr(E_0^T F_0Q) \nonumber \\
    &=& 2d - 2 \tr(E_0^TF_0Q), 
\end{eqnarray}
where $(*)$ is true by the linearity of the trace and $(**)$ is true because $Q^TF_0^TF_0Q=I_d$ and $E_0^TE_0=I_d$.
Consider the singular value decomposition $E_0^TF_0 = S \Sigma P^T$, so that $S$ and $P$ are unitary matrices of order $d$ and $\Sigma$ is a diagonal matrix of order $d$ with diagonal entries $\sigma_1\geq \sigma_2\geq\cdots\geq \sigma_d$. Since $S$ and $P$ are square matrices, $S^T$ and $P^T$ are also unitary. Let $\Tilde{Q}= PS^T$ and note that $\Tilde{Q}^T\Tilde{Q}= SP^TPS^T=SS^T=\mathbf{I}_d$. We have $\tr (\Tilde{Q} E_0^TF_0) =\tr ( E_0^TF_0\Tilde{Q})=\tr(\Sigma)= \sum_i \sigma_i $ and $ \tr(F_0^TE_0E_0^TF_0) = \tr(P \Sigma S^TS \Sigma P^T) =\tr(P \Sigma^2 P^T)  = \sum_i \sigma_i^2$. 

We wish to show that
\begin{align*}
    \norm{F_0\Tilde{Q} - E_0}_F = \inf \limits_{Q \in \mathcal{U}_d} \norm{F_0Q - E_0}_F \leq \sqrt{2}\norm{F_1^TE_0}_F.
\end{align*}
By~\eqref{eq:aux} and~\eqref{aux_eq_CH}, this is equivalent to proving that
\begin{align}\label{eq47:main}
     2d - 2 \tr(E_0^TF_0\Tilde{Q})\stackrel{(a)}=\inf \limits_{Q \in \mathcal{U}_d}  ( 2d - 2 \tr(E_0^TF_0Q))   &\stackrel{(b)}\leq 2d - 2\tr (E_0^TF_0F_0^TE_0).
\end{align}
To show the left-hand side equality $(a)$ of~\eqref{eq47:main} we prove that 
\begin{equation}\label{eq_sup}
 \tr(S \Sigma P^T\Tilde{Q}) =  \sup \limits_{Q \in\mathcal{U}_{d}} \tr(S \Sigma P^TQ) .
\end{equation}
Clearly, $\tr(S \Sigma P^T\Tilde{Q}) = \sum_{i} \sigma_{i}$. On the other hand, given $Q\in\mathcal{U}_d$, let $T=P^TQS$. Since $P$, $Q$ and $S$ are unitary matrices, $T$ is a unitary matrix. Then, $T(i,i)\leq 1$ for all $i$, so that
\begin{align*}
     \tr(S \Sigma P^TQ) &= \tr(\Sigma P^TQS)= \tr(\Sigma T)\\
     &= \sum_{i} \sigma_{i} T(i,i)\leq  \sum_{i} \sigma_{i}=\tr(S \Sigma P^T\Tilde{Q}),
 \end{align*}
establishing~\eqref{eq_sup}.

To show the right-hand side $(b)$ of~\eqref{eq47:main} we prove that 
\begin{equation*}
\tr(E_0^TF_0\Tilde{Q})=\sup \limits_{Q \in \mathcal{U}_d} \tr(E_0^TF_0Q) \geq \tr (E_0^TF_0F_0^TE_0) = \tr (F_0^TE_0E_0^TF_0).
\end{equation*}
So it suffices to show that \begin{align}
    \tr(\Sigma) \geq \tr(\Sigma^2). \label{trace-ineq}
\end{align} 
Towards~\eqref{trace-ineq}, we use the Courant-Fisher Theorem~\cite[4.2.6]{horn2012matrix} to obtain 
\begin{align*}
    |\sigma_1|^2 &= \sup _{\norm{q}=1}|(q^T P) \Sigma^2 (P^T q)| \\
    &=\sup _{\norm{q}=1}|q^T F_0^TE_0E_0^TF_0 q| \\
    % & \leq \norm{E_0E_0^T}_2 \\
    &\stackrel{(*)}\leq \sup_{\norm{x}=1}x^TE_0E_0^Tx\stackrel{(**)}\leq 1,
\end{align*}
where $(*)$ holds because $\norm{F_0q}^2_2=q^TF_0^TF_0q=q^Tq=1$ and $(**)$ holds because the eigenvalues of $E_0E_0^T$ are 0 and 1.
\end{proof}

We are now ready to prove Theorem~\ref{error_rate}. We restate it for the reader's convenience.
\addtocounter{section}{-2}
\addtocounter{theorem}{-4}
\begin{theorem}
Consider a $d$-dimensional SGBM satisfying conditions \eqref{MAdistribution} and \eqref{Mfunction} with connectivity probability functions $F_{in}$ and $F_{out}$. Let $G$ be a graph drawn from this SGBM. Let $A$ be the adjacency matrix of $G$ and let $B_\sigma$ defined in~\eqref{def:Bsigma}. 
Let $U=[u_1\cdots u_{k-1}]\in\mathbb{R}^{n\times(k-1)}$, where
$u_1,\ldots,u_{k-1}$ are orthogonal unit eigenvectors of $B_\sigma$ associated with $\lambda_* = \frac{\mu_{\text{in}}-\mu_{\text{out}}}{k}n$, and let $V=[v_1\cdots v_{k-1}]\in\mathbb{R}^{n\times(k-1)}$, where $v_1,\ldots,v_{k-1}$ are the eigenvectors of $A$ associated with the eigenvalues $\lambda'_1,\ldots,\lambda'_{k-1}$ of $A$ closest to $\lambda_ * $. For some $\epsilon>0$, the following holds a.a.s.:
\begin{align*}
    \min \limits_{Q \in\mathcal{U}_{\ell}}\norm{VQ - U}_F   \leq \frac{\sqrt{12 k^5 \log n}}{\epsilon \sqrt{n}}.
\end{align*}
\end{theorem}
\addtocounter{section}{2}
\addtocounter{theorem}{4}
% \frac{\sqrt{(k-1)\frac{4}{k}   (\mu_{in} + (k-1) \mu_{out})  \log n}}{\epsilon \sqrt{n}}
 
\begin{proof}[Proof of Theorem~\ref{error_rate}]
Recall that we are assuming that $\sigma$ is such that $1,2,\ldots,n/k$ lie in the first community, $n/k+1,n/k+2,\ldots,2n/k$ lie in the second community, and so on. We consider the following block decompositions of the adjacency matrix $A$:
\begin{align*}
    A &= \begin{bmatrix}
        A_{11} & A_{12} & \cdots & A_{1k} \\
        A_{21} & A_{22} & \cdots & A_{2k} \\
        \vdots & \vdots & \ddots & \vdots \\
        A_{k1} & A_{k2} & \cdots & A_{kk}
    \end{bmatrix}
\end{align*}
where $A_{iz}$ is the $n/k \times n/k$ submatrix of $A$ induced by the rows of $A$ associated with the points in community $i$ and the columns associated with the points in community $z$. As before, $A_{iz}(a,b)$ stands for the entry $a,b$ of the matrix $A_{iz}$ for $a,b\in\{1,\ldots,\frac{n}{k}\}$. 

Consider the random variable 
$$Y_{iz}(a) = \sum_{b= 1}^{\frac{n}{k}} A_{iz}(a,b),$$ for $i,z\in\{1,\ldots,k\}$. To produce the entries of $A$, we first randomly map the vertices $1,\ldots,n$ into $\mathbf{T}^d$, and then we draw the edges according to the functions $F_{in}$ and $F_{out}$.

We have
\begin{eqnarray*}
\mathbb{E}(Y_{ii}(a)) &=& \sum_{b = 1}^{\frac{n}{k}}\mathbb{E} A_{ii}(a,b) =  \frac{n-k}{k} \mu_{in},\\
\mathbb{E}(Y_{iz}(a)) &=& \sum_{b= 1}^{\frac{n}{k}} \mathbb{E}A_{iz}(a,b) = \frac{n}{k} \mu_{out},~\textrm{ for }i \neq z. 
\end{eqnarray*}
Moreover, for any choice of $i$ and $z$, and any $a$ in community $i$, $Y_{iz}(a)$ is the sum over $b\in\{1,\ldots,\frac{n}{k}\}$ of $A_{iz}(a,b)$, which are independent Bernoulli random variables. Then, we can apply the Chernoff bound (Lemma~\ref{chernoff}) with $\gamma=k\sqrt{\frac{6 \log n}{(n-k) \mu_{in}}}$ to obtain the following bound,
\begin{align}
\mathbb{P}\left(\left|Y_{ii}(a) - \frac{n-k}{k} \mu_{in}\right| \geq  \sqrt{6 (n-k) \mu_{in}\log n}\right) \leq \frac{2}{n^2} , \label{sspbound}
\end{align}

Similarly, for $\gamma = k\sqrt{\frac{6 \log n}{n \mu_{out}}}$ the following holds for $i\neq z$ and $a$ in community $i$:
\begin{align}
\mathbb{P}\left(\left|Y_{iz}(a) - \frac{n}{k} \mu_{out}\right| \geq  \sqrt{6 n \mu_{out}\log n} \right) \leq \frac{2}{n^2} . \label{sdpbound}
\end{align}

Let $U,V$ be as in the statement. We wish to apply Theorem \ref{Davis-Kahan}. To this end, let $M=B_\sigma$, so that $E_0=U$ is the matrix whose columns are the eigenvectors of $B_\sigma$ associated with the eigenvalue $\lambda_\ast=n(\mu_{in}-\mu_{out})/k$. Moreover, let $\tilde{M}=A$, so that $F_0=V$ is the matrix whose columns are the eigenvectors of $A$ that are associated with the eigenvalues closest to $\lambda_\ast$. Let $F_1$ be an $n\times (n-(k-1))$ so that $[F_0,F_1]$ is an orthogonal matrix of eigenvectors of $A$. By Theorem~\ref{intervalmultiplicity}, the eigenvalues of $A$ associated with the eigenvectors in $F_1$ are a.a.s.\ at distance at least $\epsilon n$ from $\lambda_\ast$. By Lemma~\ref{unitary-lemma},
\begin{equation*}
\min \limits_{Q \in\mathcal{U}_{k-1}}\norm{VQ - U}_F = \inf \limits_{Q \in\mathcal{U}_{k-1}} \norm{VQ - U}_F \leq \sqrt{2} \norm{F_1^TU}_F.
\end{equation*}
By Theorem \ref{Davis-Kahan}, for $H = A-B_\sigma$, we have
\begin{equation}\label{mainthminequa}
\sqrt{2} \norm{F_1^TU}_F \leq  \frac{\sqrt{2} \norm{F_1^T H U}_F}{\epsilon n} \stackrel{\eqref{newineq}}\leq  \sqrt{2}(k-1) \frac{\norm{HU}_F}{\epsilon n}.
\end{equation}

Let $X = HU = AU - (\mathbb{E}A)U = \begin{bmatrix}
    x_1 & x_2  & \cdots & x_{k-1}
\end{bmatrix}$. We wish to bound the $i$-th entry of the column $x_j$, which we denote by $x_j(i)$. First, given $i\in\{1,\ldots,n\}$, let $q,r\in\mathbb{Z}$, $q,r\geq0$ be such that $i = q\cdot k + r$. Fix $j\in\{1,\ldots,k-1\}$, and
let $\hat{U}=\sqrt{n} U$. By the definition of $U$ (see~\eqref{eigenvectorEA}), we may view its $j$-th column of $\hat{U}$ as a vector $\hat{u}_j$ composed of $k$ constant blocks of size $n/k$, which we denote $\hat{u}_j(1),\ldots\hat{u}_j(k)$. We have
\begin{align*}
    x_j(i) &=
    \sum_{p=1}^k \sum_{z=1}^{n/k} A_{q+1,p}(r,z) \frac{\hat{u}_{j}(p)}{\sqrt{n}} - \mathbb{E}A_{q+1,p}(r,z)\frac{\hat{u}_{j}(p)}{\sqrt{n}}
    \\
    &=
    \frac{1}{\sqrt{n}}\sum_{p=1}^k \hat{u}_{j}(p) \sum_{z=1}^{n/k} (A_{q+1,p}(r,z)-\mathbb{E}A_{q+1,p}(r,z)) 
    \\
    &=
    \frac{1}{\sqrt{n}}\sum_{p=1}^k \hat{u}_{j}(p) (Y_{q+1,p}(r)-\mathbb{E}Y_{q+1,p}(r)).
\end{align*}
Thus, we have that 
\begin{align*}
    |x_j(i)|\leq  \frac{1}{\sqrt{n}}\sum_{p=1}^k |\hat{u}_{j}(p)||Y_{q+1,p}(r)-\mathbb{E}Y_{q+1,p}(r)|\leq\frac{\sqrt{k}}{\sqrt{n}}\sum_{p=1}^k |Y_{q+1,p}(r)-\mathbb{E}Y_{q+1,p}(r)|,
\end{align*}
since $|\hat{u}_{j}(p)|<\sqrt{k}$ by the definition of $U$.

Let $\delta_{(q+1)p}=\sqrt{6  \mu_{in}\log n}$ if $q+1=p$ and $\delta_{(q+1)p}=\sqrt{6  \mu_{out}\log n}$, otherwise. For $\delta = k\sqrt{6 \log n}$, the following holds for every $q\in\{0,\ldots,k-1\}$, 
\begin{align}\label{eq:delta}
\delta \geq \sum_{p=1}^k \delta_{(q+1)p}.
\end{align}

% Now, notice that if it is true that $\frac{1}{\sqrt{n}}|Y_{q+1,p}(r)-\mathbb{E}Y_{q+1,p}(r)| > \delta_{(q+1)p}$ for every $p\in\{1,\ldots,k\}$, then of course $ |x_j(i)|  > \delta$. Then, the following is true,

For $i \in [n]$ and $j \in [k-1]$, observe that $|x_j(i)|> \delta$ only if $$\frac{1}{\sqrt{n}}|Y_{q+1,p}(r) - \mathbb{E}(Y_{q+1,p})(r)| > \delta_{(q+1)p} \textrm{ for some }p\in [k].$$

Therefore, by using the union bound over index $p\in\{1,\ldots,k\}$  and using~\eqref{sspbound} and~\eqref{sdpbound}, we have
\begin{align*}
    \mathbb{P}( |x_j(i)| > \delta ) \leq \sum_{p=1}^k  \mathbb{P}\left(\frac{1}{\sqrt{n}}|Y_{q+1,p}(r) - \mathbb{E}(Y_{q+1,p})(r)| > \delta_{(q+1)p} \right) = \frac{2k}{n^2}.
\end{align*}

Now, we have a bound for the probability of $|x_j(i)|>\delta$ for fixed $i$ and $j$, so 
by the union bound over $i$ and $j$
, we have
\begin{equation}\label{eq:prob:CH}
      \mathbb{P}( \exists i,j \text { such that } |x_j(i)| > \delta ) \leq n (k-1) \frac{2k}{n^2} =  2\frac{k(k-1)}{n}.
\end{equation}
Since $\norm{X}^2_F = \sum\limits_{j=1}^{k-1}\sum\limits_{i=1}^nx^2_j(i)=\sum\limits_{j=1}^{k-1}\norm{x_j}^2$, the following is a consequence of~\eqref{eq:prob:CH}
\begin{align*}
    \mathbb{P}\left(\norm{X}^2_F > \delta^2 n(k-1) \right)  \leq 2\frac{k (k-1)}{n}.
 \end{align*}

Of course, then we have
\begin{align*}
    \mathbb{P}\left(\norm{X}_F > \delta \sqrt{n(k-1)} \right)  \leq 2\frac{k (k-1)}{n}.
 \end{align*}
Thus, by the definition of $\delta$, with high probability,
\begin{align*}
    \norm{X}_F \leq \sqrt{6 n k^3  \log n} .
\end{align*}

Thus from~\eqref{mainthminequa}, with high probability,
\begin{align*}
    \min \limits_{Q \in\mathcal{U}_{k-1}}\norm{VQ - U}_F  \leq \frac{\sqrt{2}(k-1) \norm{X}_F}{\epsilon n} \leq  \frac{\sqrt{12 k^5 \log n}}{\epsilon \sqrt{n}}. 
\end{align*}
\end{proof}

\section{Consistency of Algorithm~\ref{Cluster_Alg} and Algorithm~\ref{Cluster_Alg2}}\label{sec:consistency}

Recall that the aim of Algorithms~\ref{Cluster_Alg} and~\ref{Cluster_Alg2} is to detect the community assignment $\sigma_n \colon [n] \rightarrow [k]$ from which an $n$-vertex random graph $G_n$ has been generated according to the SGBM. To this end, each algorithm produces its own estimator $\hat{\sigma}_n$. We say that the estimator is weakly consistent if 
$$\forall \epsilon>0, \lim_{n\rightarrow \infty}\mathbb{P}\left( \ell(\sigma_n,\hat{\sigma}_n)>\epsilon\right)=0,$$
where $\ell$ is the loss function defined in~\eqref{loss}
The estimator is strongly consistent if $$\lim_{n\rightarrow \infty}\mathbb{P}\left( \ell(\sigma_n,\hat{\sigma}_n)>0\right)=0.$$

We start with Algorithm~\ref{Cluster_Alg}. It chooses the $k-1$ eigenvectors of the adjacency matrix $A$ that are closest
to $\lambda_* = \frac{\mu_{\text{in}}-\mu_{\text{out}}}{k}n$. This defines an embedding of the $n$ vertices into $\mathbb{R}^{k-1}$, to which the algorithm applies $k$-means. Let $\mathbf{B}_{n,k}$ be the set of $n\times k$ matrices with entries in $\{0,1\}$, and let $\mathbf{P}_{n,k}=\{P\in\mathbf{B}_{n,k} :\sum_{j=1}^kP_{ij}=1\}$ be the subset containing matrices where each row contains exactly one entry equal to 1.
Given a matrix $V\in\mathbb{R}^{n\times d}$, $k$-means is a procedure that aims to find $\hat{P},\hat{X}$ such that
\begin{equation}\label{kmeansproblem}
    (\hat{P},\hat{X}) = \argmin_{\substack{P\in\mathbf{P}_{n,k} \\ X\in\mathbb{R}^{k\times d}}}\norm{PX-V}_F^2.
\end{equation}
Solving Problem~\eqref{kmeansproblem} is known to be NP-hard even for $k=2$~\cite{mahajan2012planar}. Kumar, Sabharwal, and Sen~\cite{kumar2004simple} devised a linear time $k$-means algorithm which, for some fixed $\epsilon>0$ and $k\in\mathbb{N}$, finds $(\hat{P},\hat{X})  \in \mathbf{P}_{n,k} \times \mathbb{R}^{k\times d}$ such that
\begin{equation}\label{kmeansapprox}
    \norm{\hat{P}\hat{X}-V}_F^2 \leq(1+\epsilon)\min_{\substack{P\in\mathbf{P}_{n,k} \\ X\in\mathbb{R}^{k\times d}}}\norm{PX-V}_F^2.
\end{equation}

Next we state a useful lemma that relates Theorem~\ref{error_rate} and the $(1+\epsilon)$-approximation~\eqref{kmeansapprox} of the $k$-means problem.

\begin{lemat}\label{kmeans}
    Let $\epsilon>0$, $k\geq 2$, $d\leq k$ and $V,\overline{V}\in\mathbb{R}^{n\times d}$ where $\overline{V}=\overline{P}\,\overline{X}$ with $\overline{P}\in \mathbf{P}_{n,k}$ and $\overline{X}\in\mathbb{R}^{k\times d}$. Let $(\Hat{P},\Hat{X})$ be a $(1+\epsilon)$-approximation of the $k$-means problem~\eqref{kmeansproblem} associated with $V$. Let $\sigma$ and $\hat{\sigma}$ be the community assignments induced by $\overline{P}$ and $\Hat{P}$, respectively. Let $n_{\text{min}}$ be the size of the smallest community of $\sigma$ and let $\delta = \min_{i\neq j}\norm{\overline{x}_i-\overline{x}_j}$, where $\overline{x}_i$ is the $i$-th row of $\overline{X}$. If $4(2+\epsilon)\frac{\norm{V-\overline{V}}_F^2}{\delta^2}\leq n_\text{min}$, then
    $$d^\ast_H(\hat{\sigma},\sigma)\leq 4(2+\epsilon)\frac{\norm{V-\overline{V}}_F^2}{\delta^2}.$$
\end{lemat}

The proof of Lemma~\ref{kmeans} is a slight adaptation of~\cite[Lemma 4.11]{avrachenkov2022statistical}. While \cite[Lemma 4.11]{avrachenkov2022statistical} is stated for $d=k$, our lemma is for general $d \leq k$. 

\begin{proof}
    Since the results holds for $d=k$ by~\cite[Lemma 4.11]{avrachenkov2022statistical}, let $d<k$. 
    Let $\epsilon>0$, $k\geq 2$. Let $V$, $\overline{P}$, $\overline{X}$ and $\overline{V}= \overline{P}\cdot\overline{X}$ be as in the statement. Let us consider the extended matrices $V^\ast=[v_1 v_2 \cdots v_d\text{ } \mathbf{0} \cdots \mathbf{0}]$, $\overline{X}^\ast=[\overline{x}_1 \overline{x}_2 \cdots \overline{x}_d\text{ } \mathbf{0} \cdots \mathbf{0}]$ and $\overline{V}^\ast=\overline{P}\cdot\overline{X}^\ast$ of order $n \times k$. 
    
    Let $(\Hat{P}^\ast,\Hat{X}^\ast)$ be a $(1+\epsilon)$-approximation of the $k$-means problem~\eqref{kmeansproblem} associated with $V^\ast$. Let $\Hat{X}=[\Hat{x}^\ast_1\cdots \Hat{x}^\ast_d]$ the matrix induced by the first $d$ columns of $\Hat{X}^\ast$.
    We will show that $(\Hat{P}^\ast,\Hat{X})$ is a $(1+\epsilon)$-approximation of the $k$-means problem~\eqref{kmeansproblem} associated with $V$. Since our result holds for $d=k$,
    we have
    \begin{equation}\label{eq:kmeans1}
    \norm{\hat{P}^\ast\hat{X}^\ast-V^\ast}_F^2 \leq(1+\epsilon)\min_{\substack{P\in\mathbf{P}_{n,k} \\ X\in\mathbb{R}^{k\times k}}}\norm{PX-V^\ast}_F^2.
    \end{equation}
    On the one hand, the following inequality holds: 
$$\norm{\hat{P}^\ast \Hat{X}-V}_F^2\leq\norm{\hat{P}^\ast\hat{X}^\ast-V^\ast}_F^2.$$
On the other hand, 
$$\min_{\substack{P\in\mathbf{P}_{n,k} \\ X\in\mathbb{R}^{k\times k}}}\norm{PX-V^\ast}_F^2 = \min_{\substack{P\in\mathbf{P}_{n,k} \\ X\in\mathbb{R}^{k\times d}}}\norm{PX-V}_F^2$$
So, by~\eqref{eq:kmeans1}, we have
$$\norm{\hat{P}^\ast\hat{X}-V}_F^2 \leq(1+\epsilon)\min_{\substack{P\in\mathbf{P}_{n,k} \\ X\in\mathbb{R}^{k\times d}}}\norm{PX-V}_F^2,$$
 which concludes the proof.
\end{proof}

% If $d<k$, just consider the extended matrix $V^\ast=[v_1 v_2 \cdots v_d\text{ } \mathbf{0} \cdots \mathbf{0}]$ such that we have $k$ columns. Let $X=[x_1 x_2 \cdots x_k]$ and consider $X^\ast = [x_1 x_2 \cdots x_d \text{ }\mathbf{0} \cdots \mathbf{0}]$. \color{cyan}Of course,
% $$\norm{PX^\ast-V^\ast}_F^2\leq\norm{PX-V^\ast}_F^2,$$
% which implies that the pair $(\hat{P},\hat{X})$ is such that $\hat{X}=[\hat{x}_1 \hat{x}_2 \cdots \hat{x}_d \text{ } \mathbf{0}\cdots\mathbf{0}]$. Now, applying Lemma 4.11 of \cite{avrachenkov2022statistical} we obtain a $(1+\epsilon)$ approximation $(\hat{P},\hat{X})$, such that 
% \begin{equation}
%     \norm{\hat{P}\hat{X}-V^\ast}_F^2 \leq(1+\epsilon)\min_{\substack{P\in\mathbf{P}_{n,k} \\ X\in\mathbb{R}^{k\times d}}}\norm{PX-V^\ast}_F^2.
% \end{equation}
% Let $\hat{x}_1,\ldots,\hat{x}_k$ be the columns of $\hat{X}$. Now, consider $\bar{X}=[\bar{x}_1 \bar{x}_2 \cdots \bar{x}_d]$ and $\hat{X}^\ast=[\hat{x}_1 \hat{x}_2 \cdots \hat{x}_d \text{ } \mathbf{0}\cdots\mathbf{0}]$. The Lemma~\ref{kmeans} is obtained by the pair $(\hat{P},\hat{X})$ on the following inequality
% $$ \norm{\hat{P}\bar{X}-V}_F^2=\norm{\hat{P}\hat{X}^\ast-V^\ast}_F^2\leq\norm{\hat{P}\hat{X}-V^\ast}_F^2.$$
% This shows that the pair $(\hat{P},\bar{X})$ is a $(1+\epsilon)$ approximation of the $k$-means problem.\color{black}

\begin{theorem}\label{weak_consistent}
The Algorithm \ref{Cluster_Alg} is a.a.s.\ weakly consistent for the SGBM under the hypotheses of Theorem~\ref{thm_main_formal}.
\end{theorem}
\begin{proof}

Let $\sigma\in[k]^n$ be the community assignment of the SGBM. Let $\overline{P}=(p_{ij})\in\mathbf{P}_{n,k}$ such that $p_{i\sigma_i}=1$ for each $i\in[n]$. In order to prove that the Algorithm \ref{Cluster_Alg} is weakly consistent, it should produce a node labelling $\hat{\sigma}$ such that
\begin{equation*}
    \forall \epsilon>0 : \lim_{n\rightarrow \infty}\mathbb{P}(\ell(\sigma,\hat{\sigma})>\epsilon)\rightarrow 0.
\end{equation*}

Consider $\Tilde{\sigma}$ be the node labelling obtained by Algorithm~\ref{Cluster_Alg} for a matrix $A$ drawn from the SGBM. Let $\lambda_1'\geq\cdots\geq\lambda_{k-1}'$  be the eigenvalues of $A$ closest to $\lambda_* = \frac{\mu_{\text{in}}-\mu_{\text{out}}}{k}n$. Let $v_1,\ldots,v_{k-1}$ be orthogonal unit eigenvectors of $A$ associated with $\lambda_1',\ldots,\lambda_{k-1}'$, respectively. 

Let $\Tilde{U}\in\mathbb{R}^{k\times (k-1)}$ be the matrix
\begin{align} \Tilde{U}=\label{centroids}\begin{bmatrix}
        \sqrt{\frac{k}{2n}} & \sqrt{\frac{k}{2n}} & \ldots & \sqrt{\frac{k}{2n}} \\
        -\sqrt{\frac{k}{2n}} & 0 & 0 & 0  \\
         0 &   -\sqrt{\frac{k}{2n}} & 0 & 0 \\
         \vdots &   \vdots & \vdots & \vdots \\
         0 &   0 & 0 & -\sqrt{\frac{k}{2n}} 
    \end{bmatrix},
\end{align} 
so that $\tilde{U}=\sqrt{k/n}\, U'$ for $U'$ in~\eqref{defUprime}. Consider $U=[u_2\cdots u_{k}]\in\mathbb{R}^{n\times(k-1)}$ defined as $U=\overline{P}\tilde{U}$. Note that, if $\sigma$ is the canonical assignment of the previous section, where vertices $1,\ldots,n/k$ lie in the first community, vertices $n/k+1,\ldots,2n/k$ lie in the second community, and so on, then $U$ would be the matrix defined in~\eqref{eigenvectorEA}. The columns of matrix $U$ are precisely the eigenvectors of $\mathbb{E}A$ associated with $\lambda_{*}=\frac{n}{k}(\mu_{in}-\mu_{out})-\mu_{in}$.

Now, we define $\Tilde{Q}=\text{arg}\min_{Q\in\mathcal{U}_{k-1}} \norm{VQ - U}_F$, $\overline{X} = \Tilde{U}\Tilde{Q}^T$ and $\overline{V}=\overline{P}\overline{X}$. We wish to apply Lemma~\ref{kmeans} to these matrices.  Clearly, $n_{min}=n/k$. We now find an appropriate value for $\delta=\min_{i\neq j}\norm{\overline{x}_i-\overline{x}_j}_2$, where $\overline{x}_i$ is the $i$-th row of $\overline{X}$. Let $\Tilde{u}_i$ be the $i$-th row of $\Tilde{U}$. Given $i \neq j$, we have
\begin{align*}
\norm{\overline{x}_i-\overline{x}_j}_2^2&=\norm{{\Tilde{u}}_iQ^T-\Tilde{u}_jQ^T}^2_2=\norm{({\Tilde{u}}_i-{\Tilde{u}}_j)Q^T}_2^2
% &=\tr\left((({\Tilde{u}}_i-{\Tilde{u}}_j)Q^T)^T({\Tilde{u}}_i-{\Tilde{u}}_j)Q^T)\right)\\
% &=\tr\left(Q({\Tilde{u}}_i-{\Tilde{u}}_j)^T({\Tilde{u}}_i-{\Tilde{u}}_j)Q^T)\right)\\
% &\stackrel{\eqref{trcircuit}}=\tr\left(({\Tilde{u}}_i-{\Tilde{u}}_j)^T({\Tilde{u}}_i-{\Tilde{u}}_j)Q^TQ)\right)\\
=\norm{{\Tilde{u}}_i-{\Tilde{u}}_j}_2^2\stackrel{\mbox{\tiny Lemma~\ref{distancecentroids}}}\geq\frac{k}{n}.
\end{align*}
So, $\delta^2=\min_{i\neq j}\norm{\overline{x}_i-\overline{x}_j}_2^2\geq \frac{k}{n}$. 

Next we consider
\begin{align}\label{relation4.6}
    \norm{V-\overline{V}}_F^2&=\norm{V-\overline{P}\Tilde{U}\Tilde{Q}^T}_F^2\nonumber
    \stackrel{(*)}=\norm{V\Tilde{Q}-\overline{P}\Tilde{U}\Tilde{Q}^T\Tilde{Q}}_F^2\nonumber\\
    &=\norm{V\Tilde{Q}-\overline{P}\Tilde{U}}_F^2\nonumber
    =\norm{V\Tilde{Q}-U}_F^2\nonumber\\
    &\stackrel{(**)}\leq \frac{Ck^5\log n}{n},
\end{align}
where $(*)$ comes from the fact that a multiplication of $V-P\Tilde{U}Q^T$ by a unitary matrix does not change the Frobenius norm and $(**)$ (and the constant $C$) comes from Theorem~\ref{error_rate}. Also note that $(**)$ holds a.a.s.\ (with respect to the random selection of $A$).

As a consequence, we have
$$ 4(2+\epsilon)\frac{\norm{V-\overline{V}}_F^2}{\delta^2} \leq C k^5 \log{n}\leq n_{\min}$$
for large $n$, which establishes the hypotheses of Lemma~\ref{kmeans}. Applying the lemma, we conclude that $$t=d^\ast_H(\Tilde{\sigma},\sigma) \leq 4(2+\epsilon)\frac{\norm{V-\overline{V}}_F^2}{\delta^2}\leq C k^5 \log n.$$
Since $\ell(\sigma,\hat{\sigma})=\frac{t}{n}\leq \frac{Ck^5\log n}{n}=o(1)$, Algorithm~\ref{Cluster_Alg} is a.a.s.\ weakly consistent.

\end{proof}

\begin{theorem}
The Algorithm \ref{Cluster_Alg2} is a.a.s.\ strongly consistent for the SGBM under the hypotheses of Theorem~\ref{thm_main_formal}.
\end{theorem}

\begin{proof}
Let $\sigma$ be the community assignment of the SGBM.  Let $\mu_{in}>\mu_{out}$ and $\epsilon=(\mu_{in}-\mu_{out})/3$. 
Let $\mathcal{C}$ be the property of a node having at least $n(\mu_{in}-\epsilon)/k$ neighbors within its own community and at most $n(\mu_{out}+\epsilon)/k$ neighbors in each of the other communities.
We first show that a matrix $A$ drawn according to the SGBM a.a.s.\ satisfies the property $\mathcal{C}$ for every node $v$.
% We first show that a matrix $A$ drawn according to the SGBM a.a.s.\ satisfies the following property: every vertex $v$ has at least $n(\mu_{in}-\epsilon)/k$ neighbors within its own community and at most $n(\mu_{out}+\epsilon)/k$ neighbors in each of the other communities. 

Indeed, let $N_{\ell}(v)$ be the number of neighbors of $v$ in the community as $\ell$. For $\gamma = \epsilon \mu_{in}$, by Lemma~\ref{chernoff} we have that
$$\mathbb{P}\left(\left|
N_{\sigma_v}(v) - \mu_{in} \frac{n}{k}
\right|\geq \epsilon \frac{n}{k} \right)\leq 2\exp\left(-\frac{\epsilon^2\mu_{in}^3n}{3k}\right).$$ And, given $\ell\neq\sigma_v$ another community label, for $\gamma = \epsilon \mu_{out}$ we have that
$$\mathbb{P}\left(\left|
N_{\ell}(v) - \mu_{out} \frac{n}{k}
\right|\geq \epsilon \frac{n}{k} \right)\leq 2\exp\left(-\frac{\epsilon^2\mu_{out}^3n}{3k}\right).$$
So, using the union bound over $\ell\in\{1,\ldots, k\}$, we have that
$$\mathbb{P}\left(v \mbox{ does not satisfies }\mathcal{C} \right)\leq2k\exp\left(-\frac{\epsilon^2\mu_{out}^3n}{3k}\right).$$
Finally, using the union bound over $v\in\{1,\ldots,n\}$, we have that 
$$\mathbb{P}\left(\exists v : v\mbox{ does not satisfies }\mathcal{C} \right)\leq2kn\exp\left(-\frac{\epsilon^2\mu_{out}^3n}{3k}\right),$$
which goes to $0$ as $n$ goes to infinity.

To prove our statement, we will show that every vertex $v$ that satisfies the conditions of the above paragraph is classified correctly by Algorithm \ref{Cluster_Alg2}. \color{black} Let $\tilde{\sigma}$ be the community label obtained by Algorithm \ref{Cluster_Alg} applied to $A$.
Let 
$$\Tilde{Z_j}(v) = \sum\limits_{q \in D \colon \Tilde{\sigma}_q=j} A(i,q), \quad \textrm{ and } \quad Z_j(v) = \sum\limits_{q \in D \colon \sigma_q=j} A(i,q) \textrm{ for }j \in [k],$$
that is, $Z_j(v)$ is the number of neighbors of $v$ in community $j$ with respect to $\sigma$, and $\Tilde{Z_j}(v)$ is the number of neighbors of $v$ in community $j$ with respect to $\tilde{\sigma}$. 
By the previous paragraph and by our choice of $\epsilon$, we know that a.a.s.\ the following holds for all $v$ and for all $j \neq \sigma_v$:
\begin{equation}\label{aux_eq}Z_{\sigma_v}(v)\geq \frac{n}{k}(\mu_{in}-\epsilon), \quad \textrm{ and }\quad ~\frac{n}{k}(\mu_{out}+\epsilon)>Z_{j}(v).
\end{equation}

By Theorem \ref{weak_consistent}, the total number of possible $q$ such that $\Tilde{\sigma}_q\neq\sigma_q$ is a.a.s.\ bounded by $C\log n$, for some $C>0$. Then, for all $v\in [n]$ and all $j\in[k]$, we have
\begin{align}\label{abovebound}
|Z_j(v) - \tilde{Z}_j(v)| \leq C k^5 \log n.
\end{align}
For large $n$, this means that, a.a.s. for all $j \neq \sigma_v$,
\begin{eqnarray*}
\tilde{Z}_{\sigma_v}(v)-\tilde{Z}_{j}(v)&\geq& {Z}_{\sigma_v}(v)-{Z}_{j}(v)- 2 C k^5 \log n\\
&\stackrel{\eqref{aux_eq}}{\geq}& \frac{n(\mu_{in}-\mu_{out})}{3k} - 2 C k^5 \log n >0.
\end{eqnarray*}
In other words, $v$ is assigned to cluster $\sigma_v$ by Algorithm~\ref{Cluster_Alg2}.
\end{proof}

\section{Final remarks}\label{sec:simulation}

In this paper, we extended the community detection algorithm~\cite{avrachenkov2022} that applies to the Soft Geometric Block Model for two communities to an arbitrary number $k \geq 2$ of communities. While the algorithm for two clusters relied on singling out a particular eigenvalue and its associated eigenvector, the generalization uses a vector space spanned by $k-1$ eigenvectors, for which the structure is inherently more delicate. The basis of the eigenspace is no longer uniquely determined, and the new algorithm uses this basis to produce an embedding into $\mathbb{R}^{k-1}$. In fact, the algorithm uses a new additional step of applying $k$-means to this embedding, and new arguments are needed to analyze this part.

A significant part of the technical challenge lies in controlling the behavior of eigenvectors under perturbations and ensuring that their geometric configuration remains sufficiently stable to allow clustering via $k$-means. To this end, we rely on a nontrivial application of the Davis–Kahan Theorem and develop auxiliary results in matrix theory that may be of independent interest. These tools were important to fill the gap between the expected spectral structure and the empirical spectral embedding derived from the adjacency matrix. Our results provide a theoretical foundation for spectral methods in geometric random graphs with multiple communities, but also open up a number of natural questions for future work:
\begin{enumerate}
	\item[(1)] What happens when the technical conditions of the theorem fail? For instance, can we extract any information if there is $z\in \mathbb{Z}_{d}$ such that \eqref{eq:24} does not hold?
	\item[(2)] Can the algorithm be applied to the SGBM in cases where the communities are not of equal size? For instance, to cases where the sizes of the communities are part of the input, or where each element is assigned u.a.r.\ to one of the communities.
    \item[(3)]	What would happen if, instead of depending on two functions $F_{in}$ and $F_{out}$ that govern intra-community and inter-community connections, respectively, the connectivity function $F$ depended on functions $F_{ij}$ that govern the connections between members of communities $i$ and $j$, for each pair $(i,j)\in [k]^2$?
	\item[(4)] Can we soften the condition that the elements are embedded into $\mathbf{T}^d$ u.a.r.?  Could the analysis be adapted to other probability distributions on $\mathbf{T}^d$ or to metric spaces other than the torus?
	\item[(5)]	How does the algorithm behave in the sparse regime, where the average degree is sublinear?
	% \item[(6)]	Is there a principled, data-driven way to identify the relevant spectral window without prior knowledge of $\mu_{\text{in}}$ and $\mu_{\text{out}}$?
\end{enumerate}

\section*{Acknowledgments}
Allem, Hoppen and Sibemberg acknowledge the partial support by
CAPES under project MATH-AMSUD 88881.694479/2022-01. They also acknowledge partial support by CNPq (Proj.\ 408180/2023-4). L.\ E.\ Allem was partially supported by FAPERGS 21/2551-
0002053-9. C.~Hoppen was partially supported of CNPq (Proj.\ 315132/2021-3). CAPES is Coordena\c{c}\~{a}o de Aperfei\c{c}oamento de Pessoal de N\'{i}vel Superior. CNPq is Conselho Nacional de Desenvolvimento Cient\'{i}fico e Tecnol\'{o}gico. FAPERGS is Funda\c{c}\~{a}o de Amparo \`{a} Pesquisa do Estado do Rio Grande do Sul.

\bibliographystyle{amsplain}
% \bibliography{references}
\providecommand{\bysame}{\leavevmode\hbox to3em{\hrulefill}\thinspace}
\providecommand{\MR}{\relax\ifhmode\unskip\space\fi MR }
\providecommand{\MRhref}[2]{%
  \href{http://www.ams.org/mathscinet-getitem?mr=#1}{#2}
}
\providecommand{\href}[2]{#2}

\appendix
\section{Auxiliary Results}

\begin{theorem}\label{combinat}
    Let $b = (b_1,\ldots,b_m)$ be a binary tuple and let $X^m_k=\{(x_1,\ldots,x_m):x_i\in\{d_1,\ldots,d_k\}\}$ be the set of $k$-ary tuples with size $m$, where $d_1,\ldots,d_k$ are the possible digits of the $k$-ary tuple. Consider the set $S_k(b) = \{x\in X_k^m: x_i=x_{i+1}\text{ if } b_i=0 \text{, and } x_i \neq x_{i+1}\text{ if } b_i=1 \text{ for } i\in[m]\}$, where we write $x_{m+1}=x_1$. Let $B^m_p=\{(b_1,\ldots,b_m):\sum_{i=1}^mb_i=p\}$, for $p\in\{0,\ldots,m\}$. We have
    \begin{eqnarray*}
     \sum_{b\in B_p^m} \left|S_k(b)\right| &=&
    \begin{cases}\binom{m}{p}((k-1)^p + (k-1)) \text{ if } p \text{ is even,} \\
     \binom{m}{p}((k-1)^p - (k-1)) \text{ if } p \text{ is odd.}
    \end{cases} 
    \end{eqnarray*}

    % Given all the binary number $X_i$ with $p$ ones, there are $\binom{m}{p}((k-1)^p + (k-1))$ number of $k$-ary numbers in $\Sigma_i$ when $p$ is even, and $\binom{m}{p}((k-1)^p - (k-1))$ number of  $k$-ary numbers in $\Sigma_i$ when $p$ is odd. 
\end{theorem}
\begin{proof}
    Let us first clarify the definitions by a small example in Figure~\ref{tree1}. Consider $m=3$, $p=2$ and $k=3$. We have that $(1,0,1),$ $ (1,1,0)$ and $(0,1,1)$ are in $B^m_p$. Let $b=(1,0,1)$ and  $X^3_3=\{(x_1,x_2,x_3):x_i\in\{d_1,d_2,d_3\}\}$ be the set of ternary tuples formed with $d_1,d_2,d_3$. We will calculate the size of $S_3(b)$. First, as shown in Figure~\ref{tree1}, fix $x_1=d_1$. Since, $b_1=1$, then $x_2$ must be different from $x_1$. Then for $x_2$, there are two possibilities $d_2$ or $d_3$. Since $b_2=0$, $x_3=x_2$. Finally, since we have $b_3=1$ as the third binary digit, we have two possibilities different from $x_3$. However, for the valid ternary tuples the only possibility is $x_1=d_1$, as it should finish at the same digit the tree started.  Given that the choice of $x_1=d_1$ was arbitrary, we have $|S_3(b)|=6$. Also, we have that $|S_3(1,1,0)| =|S_3(0,1,1)| = 6$. Finally, 
    $$ \sum_{b\in B_2^3} \left|S_3(b)\right| = 18 .$$
    
    % {\color{red}And hence the total number of possible tuples is six. btw, do we obtain this from the formula?}

\begin{figure}
\centering
    \begin{tikzpicture}[level distance=1.5cm,
  level 1/.style={sibling distance=3cm},
  level 2/.style={sibling distance=1.5cm}]
  \node {$d_1$}
    child {node {$d_2$}
      child{ node{$d_2$} {child {node {$d_1$}}
      child {node {$d_3$}}}}
    }
    child {node {$d_3$}
    child{ node {$d_3$} { 
    child {node {$d_1$}}
      child {node {$d_2$}}
    }}};
\end{tikzpicture}\caption{Ternary tree starting with $d_1$, corresponding to binary digits $101$.}\label{tree1}
\end{figure}
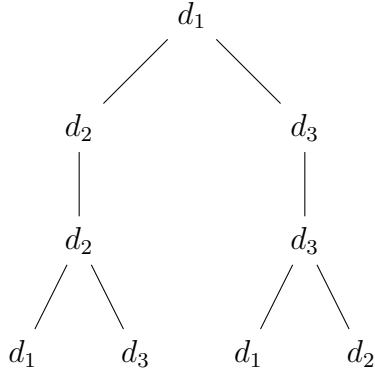

Now we proceed with the formal proof. Fix a vector $b\in B_p^m$.
% Let $(x_1,x_2, \ldots, x_m)$ {\color{red} be an element in $S_k(b)$. 
Assume $x_1=d_1$ has been fixed. Given $x_1$ and $b$ we define a $k$-ary tree $T$, which has $m+1$ layers. We start with a root which is labeled by an element of $
\{d_1,\ldots,d_k\}$, say $d_1$. To define the next layer we consider $b_1$. If $b_1=0$, we connect the root to a single child and with the same label. If $b_1=1$, we connect it to $k-1$ children, each with one of the labels that is different from their parent. Now, suppose we already defined the layer $l$ of $k$-ary tree. If $b_l=0$, each vertex of layer $l$ has a single child, which keeps the same label. If $b_l=1$, each vertex of layer $l$ has $k-1$ children, one with each of the other labels. For $k=m=3$, $b=(1,0,1)$ and $x_1=d_1$, the ternary tree is shown in Figure~\ref{tree1} as an example. There is a bijection between elements of $S_k(b)$ and paths from the root of the tree to leaves of the tree whose label coincide with the root's label. These are called valid paths.
%be the starting digit of the number we are constructing.

Given a digit $d_i$, let $n_l(d_i)$ denote the number of occurrences of $d_i$ in level $l$ of the $k$-ary tree $T$. By definition, we have $n_{l+1}(x_i) = n_l(x_i)$ if $b_l = 0$ and $n_{l+1}(x_i) = \sum_{j \neq i} n_l(x_j)$, otherwise.
% \begin{align*}
%   n_{l+1}(x_i) = n_l(x_i)  \quad  \text{if} \quad  b_l = 0 \\
%    n_{l+1}(x_i) = \sum_{j \neq i} n_l(x_j)  \quad  \text{if} \quad  b_l = 1. 
% \end{align*}
We define the vector $\mathbf{y}_l = \mathbf{y}_{l}(x_1,b) = \begin{bmatrix} n_{l}(d_1) & n_{l}(d_2) & \cdots & n_l(d_k ) \end{bmatrix}^T $, so that \begin{align*}
 \mathbf{y}_l = (\mathbf{1}\mathbf{1}^T - \mathbf{I}_k)\mathbf{y}_{l-1}  \quad \text{if} \quad b_{l-1} = 1 \\
  \mathbf{y}_l = \mathbf{I}_k \mathbf{y}_{l-1}  \quad \text{if} \quad b_{l-1} = 0 .
\end{align*}

Since there are $p$ occurrences of $1$ in $b$, this immediately leads to
\begin{align}\label{eq:a1}
    \mathbf{y}_l = (\mathbf{1}\mathbf{1}^T - \mathbf{I}_k)^p \mathbf{y}_0. 
\end{align}
To solve~\eqref{eq:a1}, we write 
% To compute the power of $(\mathbf{1}\mathbf{1}^T - \mathbf{I}_k)$, we make use of the recurrence relation in the coefficients of $\mathbf{1}\mathbf{1}^T$ and $I$. Let 
% \begin{align*}
    $(\mathbf{1}\mathbf{1}^T - \mathbf{I}_k)^s = \alpha_s \mathbf{1}\mathbf{1}^T  + \beta_s \mathbf{I}_k,$ 
% \end{align*}
so that
\begin{align*}
    (\mathbf{1}\mathbf{1}^T - \mathbf{I}_k)^{s+1} = ((k-1) \alpha_s + \beta_s)\mathbf{1}\mathbf{1}^T  - \beta_s \mathbf{I}_k. 
\end{align*}
From this, we get
% These coefficients in the matrix form is obtained by
% \begin{align*}
%    \begin{bmatrix} a_{s+1} \\b_{s+1} \end{bmatrix} = \begin{bmatrix}
%        k-1 & 1 \\ 0 & -1 \end{bmatrix} \begin{bmatrix} a_{s} \\b_{s} \end{bmatrix}.
% \end{align*}
% Using the {\color{red}Eigendecomposition} of the coefficient matrix, we get
\begin{align*}
   \begin{bmatrix} \alpha_{s+1} \\ \beta_{s+1} \end{bmatrix} = \begin{bmatrix}
       k-1 & 1 \\ 0 & -1 \end{bmatrix} \begin{bmatrix} \alpha_{s} \\ \beta_{s} \end{bmatrix} = \begin{bmatrix}
       1 & \frac{-1}{k} \\ 0 & 1 \end{bmatrix}  \begin{bmatrix}
       (k-1)^s & 0 \\ 0 & (-1)^s \end{bmatrix} \begin{bmatrix}
       1 & \frac{1}{k} \\ 0 & 1 \end{bmatrix}\begin{bmatrix} 1 \\-1 \end{bmatrix}.
\end{align*}
It follows that 
\begin{align*}
     \begin{bmatrix} \alpha_{s+1} \\ \beta_{s+1} \end{bmatrix} =  \begin{bmatrix} \frac{(k-1)^{s+1} + (-1)^{s}}{k} \\ (-1)^{s+1}
     \end{bmatrix}.
\end{align*}

If we assume that $x_1=d_1$, we have $\mathbf{y}_0=\mathbf{e}_1$,
the canonical basis vector. First consider the case where $b_m=1$. The number of valid paths is given by
\begin{align*}
    N_1 &=  (\mathbf{1}- e_1) ^T(\mathbf{1}\mathbf{1}^T - \mathbf{I}_k)^{p-1} e_1 
% & = (\mathbf{1}- e_1)^T (a_{p-1}  \mathbf{1}\mathbf{1}^T  + b_{p-1} \mathbf{I}_k)e_1 \\
% &= (k-1) a_{p-1} \\
= (k-1)\frac{(k-1)^{p-1} + (-1)^{p}}{k}.
\end{align*}
If $b_m=0$, the number of valid paths is
\begin{align*}
    N_0 &=  (e_1)^T (\mathbf{1}\mathbf{1}^T - \mathbf{I}_k)^{p} e_1 
 = \frac{(k-1)^{p} + (-1)^{p-1}}{k} + (-1)^{p}=\frac{(k-1)^{p} + (-1)^{p}(k-1)}{k}=N_1.  
\end{align*}
Thus, since $N_0=N_1$, the number of paths is the same for any vector $b\in B_p^m$. As we can start with any of the $k$ digits $d_1,\ldots,d_k$, we have $$|S_k(b)| = (k-1)^{p} + (-1)^{p}(k-1).$$ This completes the proof.
\end{proof}

In the next lemma, we use the standard notation
$$\sinc{x}=\begin{cases} 
\frac{\sin{x}}{x},& \textrm{ if }x\neq 0,\\
1,& \textrm{ if }x=0.
\end{cases}$$
\begin{lemat}\label{append_sbm}
Let $p\in(0,1]$. Let $F:\mathbf{T}^d\rightarrow \mathbb{R}$ be the constant function such that $F(x) = p$. For all $z\in\mathbb{Z}^d$, we have
$$\Hat{F}(z) = p\prod_{j=1}^k \sinc(\pi z_j).$$
\end{lemat}

\begin{proof}
Let $z\in\mathbb{Z}^d$ and consider
    \begin{eqnarray}
        \Hat{F}(z) &=& \int_{\left[-\frac{1}{2},\frac{1}{2}\right]^d}F(x)e^{-2i\pi \langle z,x\rangle}dx\nonumber \\
        &=& \int_{-\frac12}^{\frac12} \cdots \int_{-\frac12}^{\frac12}pe^{-2i\pi (z_1x_1+\cdots+z_dx_d)}dx_1\ldots dx_d\nonumber\\
        &=&p\prod_{j=1}^{d}\int_{-\frac12}^{\frac12}e^{-2i\pi  z_jx_j}dx_j\label{inte_coeff}.
    \end{eqnarray}
For any $j\in\{1,\ldots,d\}$ such that $z_j\neq 0$, we have
\begin{eqnarray}
    \int_{-\frac{1}{2}}^{\frac{1}{2}}e^{-2i\pi  z_jx_j}dx_j
    &=&\frac{e^{-2\pi iz_j \frac{1}{2}}-e^{2\pi iz_j \frac{1}{2}}}{-2\pi iz_j \frac{1}{2}}\nonumber\\
    &=&\frac{\sin (\pi z_j)}{\pi z_j}=\sinc (\pi z_j).\label{inte_part}
\end{eqnarray}
The result follows from~\eqref{inte_coeff} because, for $z_j=0$, 
$$\int_{-\frac{1}{2}}^{\frac{1}{2}}e^{-2i\pi  z_jx_j}dx_j=\int_{-\frac{1}{2}}^{\frac{1}{2}}dx_j=1.$$
\end{proof}

\begin{lemat}\label{new_condition_gbm}
Consider the $d$-dimensional GBM model, where $F_{\text{in}}, F_{\text{out}}$ are $1$-periodic, and defined on the flat torus $\mathbf{T}^d$ by $F_{\text{in}}(x) = 1(\|x\| \leq r_{\text{in}})$ and $F_{\text{out}}(x) = 1(\|x\| \leq r_{\text{out}})$, with $r_{\text{in}} > r_{\text{out}} > 0$. Denote by $\mathcal{B}$ the set of parameters $r_{\text{in}}$ and $r_{\text{out}}$ defined by negation of conditions \eqref{eq:24} and \eqref{eq:25}:
$$
\mathcal{B} = \left\{ (r_{\text{in}}, r_{\text{out}}) \in \mathbb{R}^2_{+} : \widehat{F}_{\text{in}}(z) + (k-1)\widehat{F}_{\text{out}}(z) = \mu_{\text{in}} - \mu_{\text{out}} \text{ for some } z \in \mathbb{Z}^d \right\}
$$
Then the set of `bad' parameters is of zero Lebesgue measure, that is, $\text{Leb}(\mathcal{B}) = 0$.

\end{lemat}

\begin{proof}
This proof is just an adaptation of the proof of Proposition 2~\cite{avrachenkov2022}, so we state this adaptations. By Lemma 3 of the Appendix of~\cite{avrachenkov2022}, proving that $(r_{in},r_{out})\in\mathcal{B}$ is the same as proving that, given $z\in\mathbb{Z}^d$
\begin{equation}
r_{\text{in}}^{d} \prod_{j=1}^{d} \text{sinc} \left( 2 \pi r_{\text{in}} z_{j} \right) 
+ (k-1) r_{\text{out}}^{d}  \prod_{j=1}^{d} \text{sinc} \left( 2 \pi r_{\text{out}} z_{j} \right)
= r_{\text{in}}^{d} - r_{\text{out}}^{d},
\end{equation}
So we define,
$$
f_{z}(x) = x^{d} \left( 1 + (k-1)\prod_{j=1}^{d} \text{sinc} \left( 2 \pi x z_{j} \right) \right),
$$
$$
g_{z}(x) = x^{d} \left( 1 - \prod_{j=1}^{d} \text{sinc} \left( 2 \pi x z_{j} \right) \right),
$$
for some $z = (z_{1}, \dots, z_{d}) \in \mathbb{Z}^{d}$. Now, just consider $h_z:\mathbb{C}\rightarrow \mathbb{R}$, such that $$h_z(y)=y^d \left( 1 + (k-1)\prod_{j=1}^{d} \text{sinc} \left( 2 \pi y z_{j} \right) \right).$$ As in Lemma 3~\cite[ Appendix]{avrachenkov2022} $h_z$ is holomorphic. This implies that $h'_k(y)$ is holomorphic, so it has a countable number of $0$. This also implies that $f_z$ has countable many zeros, since $h_k'\equiv f'_k$ in $\mathbb{R}$. The rest of the proof now goes exactly like~\cite{avrachenkov2022}. 
\end{proof}

\end{document}